\documentclass[journal]{IEEEtran}
\usepackage{amsmath}
\usepackage{amssymb}

\usepackage{algorithm}
\usepackage{algpseudocode}
\usepackage{booktabs}
\usepackage{subcaption}
\usepackage[short]{optidef}
\usepackage{easy-todo}
\usepackage{gensymb}
\usepackage{soul,xcolor}
\usepackage{siunitx}
\usepackage{amsthm}
\theoremstyle{plain}
\newtheorem{remark}{Remark}

\newtheorem{definition}{Definition}
\newtheorem{lemma}{Lemma}

\usepackage{pgfplots}
\usepgfplotslibrary{external}
\usetikzlibrary{pgfplots.groupplots,positioning}
\pgfplotsset{
    compat=newest,
	tick label style={font=\small},
	label style={font=\small},
	legend style={font=\footnotesize},
    /pgfplots/legend image code/.code={%
        \draw[mark repeat=3,mark phase=2,#1]
            plot coordinates {
                (0cm,0cm)
                (0.2cm,0cm)
                (0.4cm,0cm)
            };
    },
}
\newlength\figureheight
\newlength\figurewidth
\newcommand{\eqd}{ \, .}
\newcommand{\eqc}{ \, ,}
\newcommand{\set}[1]{\mathcal{#1}}

\newcommand{\nt}{M}
\newcommand{\nd}{{n_d}}
\newcommand{\nz}{{n_z}}
\newcommand{\setst}[2]{\left\{#1 \, \middle| #2 \right\}}
\DeclareMathOperator{\E}{\mathbb{E}}
\DeclareMathOperator{\tr}{\text{tr}}

\DeclareMathOperator{\diag}{\text{diag}}
\newcommand{\tp}{{\mkern-1.5mu\mathsf{T}}}
\renewcommand{\IEEEPARstart}{}
\begin{document}
%==============================================================

%==============================================================
%TITLE PAGE
%==============================================================
\title{\vspace{-0.01 cm}Cautious Model Predictive Control using Gaussian Process Regression}
\author{Lukas~Hewing, Juraj~Kabzan, 
        Melanie~N.~Zeilinger% <-this % stops a space
% \thanks{The authors are with the Institute for Dynamic Systems and Control, ETH Z\"urich, Z\"urich
% CH-8092, Switzerland (e-mail: lhewing | kabzanj | mzeilinger@ethz.ch)}
\thanks{All authors are with the Institute for Dynamic Systems and Control, ETH Z\"urich.
        {\tt\footnotesize [lhewing|kabzanj|mzeilinger]@ethz.ch}}%
\thanks{This work was supported by the Swiss National Science Foundation under grant no. PP00P2 157601 / 1.}}

\maketitle

%==============================================================

\begin{abstract}
% !TeX root = gpmpc.tex
Gaussian process (GP) regression has been widely used in supervised machine
learning due to its flexibility and inherent ability to describe uncertainty in
function estimation. In the context of control, it is seeing increasing use for
modeling of nonlinear dynamical systems from data, as it allows the direct
assessment of residual model uncertainty. We present a model predictive control
(MPC) approach that integrates a nominal system with an additive nonlinear part
of the dynamics modeled as a GP\@. Approximation techniques for propagating the
state distribution are reviewed and we describe a principled way of formulating
the chance constrained MPC problem, which takes into account residual
uncertainties provided by the GP model to enable cautious control. Using
additional approximations for efficient computation, we finally
demonstrate the approach in a simulation example, as well as in a hardware
implementation for autonomous racing of remote controlled race cars,
highlighting improvements with regard to both performance and safety over a
nominal controller.
\end{abstract}

\begin{IEEEkeywords}
	Model Predictive Control, Gaussian Processes, Learning-based Control
\end{IEEEkeywords}

% !TeX root = gpmpc.tex
%%%%%%%%%%%%%%%%%%%%%%%%%%%%%%%%%%%%%%%%%%%%%%%%%%%%%%%%%%%%%%%%%%%%%%%%%%%%%%%%
\section{Introduction}\label{sc:introduction}
%%%%%%%%%%%%%%%%%%%%%%%%%%%%%%%%%%%%%%%%%%%%%%%%%%%%%%%%%%%%%%%%%%%%%%%%%%%%%%%%
Many modern control approaches depend on accurate model descriptions to enable
safe and high performance control. Identifying these models, especially for
highly nonlinear systems, is a time-consuming and complex endeavor. Often times,
however, it is possible to derive an approximate system model, e.g.\ a
linear model with adequate accuracy close to some operating point, or a simple
model description from first principles. In addition, measurement data from 
previous experiments or during operation is often available, which
can be exploited to enhance the system model and controller performance.
In this paper, we present a model
predictive control (MPC) approach, which improves such a nominal model
description from data using Gaussian Processes (GPs) to safely enhance
performance of the system. 

Learning methods for automatically identifying dynamical models from data have
gained significant attention in the past years~\cite{Hunt1992} and
more recently also for robotic applications~\cite{Nguyen-Tuong2011,Sigaud2011}.
In particular, nonparametric methods, which
have seen wide success in machine learning, offer significant potential for
control~\cite{Pillonetto2014}. The appeal of using Gaussian Process regression
for model learning stems from the fact that it requires little prior process
knowledge and directly provides a measure of residual model uncertainty. 
%This
%paper shows how to make use of GP models and their uncertainty measure for
%cautious model predictive control.
In predictive control, GPs were successfully applied to improve control
performance when learning periodic time-varying
disturbances~\cite{Klenske2016a}. The task of learning the system dynamics as
opposed to disturbances, however, imposes the challenge of propagating
probability distributions of the state over the prediction horizon of the
controller. Efficient methods to evaluate GPs for Gaussian inputs have been
developed in~\cite{Quinonero-Candela2003,Kuss2006,Deisenroth2010}. By
successively employing these approximations at each prediction time step,
predictive control with GP dynamics has been presented in~\cite{Kocijan2004}.
In~\cite{Grancharova2008} a piecewise linear approximate explicit solution for
the MPC problem of a combustion plant was presented. Application of a one-step
MPC with a GP model to a mechatronic system was demonstrated in~\cite{Cao2016}
and the use for fault-tolerant MPC was presented in~\cite{Yang2015}. A
constrained tracking MPC for robotic applications, which uses a GP to improve
the dynamics model from measurement data, was shown in~\cite{Ostafew2016a}. A
variant of this approach was realized in~\cite{Ostafew2016}, where uncertainty
of the GP prediction is robustly taken into account using confidence bounds.
Solutions of GP-based MPC problems using sequential quadratic programming
schemes is discussed in~\cite{Cao2017}. A simulation study for the application
of GP-based predictive control to a drinking water network was presented
in~\cite{Wang2016} and data efficiency of these formulations for learning
controllers was recently demonstrated in~\cite{Kamthe2018}. 

The goal of this paper is both to provide an overview of existing techniques and
to propose extensions for deriving a systematic and efficiently solvable
approximate formulation of the MPC problem with a GP model, which incorporates
constraints and takes into account the model uncertainty for cautious control.
% Although the concept of combining GP models
% with MPC has therefore been addressed in previous work, practical aspects of
% posing and solving the resulting optimization problem are intricate and require
% a combination of various results. 
Unlike most previous approaches, we specifically consider the combination of a
nominal system description with an additive GP part which can be of different
dimensionality as the nominal model. This corresponds to a common scenario in
controller design, where an approximate nominal model is known and can be
improved on, offering a multitude of advantages. A nominal model allows for
rudimentary system operation and the collection of measurement data, as well as
the design of pre-stabilizing ancillary controllers which reduce the state
uncertainty in prediction. Additionally, it allows for learning only specific
effects from data, potentially reducing the dimensionality of the machine
learning task. For example, most dynamical models derived from physical laws
include an integrator chain, which can be directly represented in the nominal
dynamics and for which no nonlinear uncertainty model has to be learned from
data---it would rather add unnecessary conservatism to the controller. The
separation between system and GP model dimension is therefore key for the
application of the GP-based predictive control scheme to higher order systems.

The paper makes the following contributions. A review and compact summary of
approximation techniques for propagating GP dynamics and uncertainties is
provided and extended to the additive combination with nominal dynamics. The
approximate propagation enables a principled formulation of chance constraints
on the resulting state distributions in terms of probabilistic reachable
sets~\cite{Hewing2018b}. In addition, the nominal system description allows for
reductin the GP model learning to a subspace of states and inputs. We discuss
sparse GPs and a tailored selection of inducing points for MPC to reduce the
computational burden of the approach. The use of some or all of these techniques
is crucial to make the computationally expensive GP approach practically
feasible for control with sampling times in the millisecond range. We finally
present two application examples. The first is a simulation of an autonomous
underwater vehicle, illustrating key concepts and advantages in a simplified
setting. Second, we present a hardware implementation for autonomous racing of
remote controlled cars, showing the real-world feasibility of the approach for
complex high performance control tasks. To the best of our knowledge this is the
first hardware implementation of a Gaussian Process-based predictive control
scheme to a system of this complexity at sampling times of \SI{20}{ms}.

% !TeX root = gpmpc.tex
%%%%%%%%%%%%%%%%%%%%%%%%%%%%%%%%%%%%%%%%%%%%%%%%%%%%%%%%%%%%%%%%%%%%%%%%%%%%%%%%
\section{Preliminaries}\label{sc:preliminaries}
%%%%%%%%%%%%%%%%%%%%%%%%%%%%%%%%%%%%%%%%%%%%%%%%%%%%%%%%%%%%%%%%%%%%%%%%%%%%%%%%
%===============================================================================
\subsection{Notation}
%===============================================================================
%
The $i$-th element of a vector $x$ is denoted ${[x]}_i$. Similarly,
${[M]}_{ij}$ denotes element $ij$ of a matrix $M$, and ${[M]}_{\cdot i}$,
${[M]}_{i\cdot}$ its $i$-th column or row, respectively. We use $\text{diag}(x)$
to refer to a diagonal matrix with entries given by the vector $x$. 
The squared Euclidean norm weighted by $M$, i.e.\ $x^\tp Mx$
is denoted $\Vert x \Vert^2_M$. We use boldface to emphasize stacked quantities,
e.g.\ a collection of vector-valued data in matrix form. $\nabla \!  \! f(z)$ is
the gradient of $f$ evaluated at $z$. The Pontryagin set difference is denoted
$\set{A} \ominus \set{B} = \setst{a}{a+b \in \set{A} \ \forall b \in \set{B}}$.
A normally distributed vector $x$ with mean $\mu$ and variance $\Sigma$ is given
by $x \sim \mathcal{N}(\mu,\Sigma)$. The expected value of $x$ is
$\mathbb{E}(x)$, $\text{cov}(x,y)$ is the covariance between vectors $x$ and
$y$, and the variance $\text{var}(x) = \text{cov}(x,x)$. We use $p(x)$,
($p(x|y)$) to refer to the (conditional) probability densities of $x$. Similary
$\Pr(E)$, ($\Pr(E \,|\, A)$) denotes the probability of an event $E$ (given
$A$).

Realized quantities during closed loop control are time indexed using
parenthesis, i.e. $x(k)$, while quantities in prediction use subscripts, e.g.
$x_i$ is the predicted state $i$-steps ahead.
%
%
%===============================================================================
\subsection{Problem formulation}\label{ssc:problem_formulation}
%===============================================================================
%
We consider the control of dynamical systems that can be represented by the
discrete-time model
\begin{align} \label{eq:system}
	x(k\!+\!1) = f(x(k),u(k)) \! + \! B_d (g(x(k),u(k)) \! + \! w(k)) \eqc
\end{align}
where $x(k) \in \mathbb{R}^{n_x}$ is the system state and $u(k) \in
\mathbb{R}^{n_u}$ the control inputs at time $k$. The model is composed of a
known nominal part $f$ and additive dynamics $g$ that describe initially unknown
dynamics of the system, which are to be learned from data and are assumed to lie
in the subspace spanned by $B_d$. We consider i.i.d.\ process noise $w_k \sim
\mathcal{N}\left(0, \Sigma^w \right)$, which is spatially uncorrelated, i.e.\
has diagonal variance matrix $\Sigma^w = \diag([\sigma^2_1, \ldots,
\sigma^2_\nd])$. We assume that both $f$ and $g$ are differentiable functions.

The system is subject to state and input constraints $\set{X} \subseteq
	\mathbb{R}^{n_x}$, $\set{U} \subseteq \mathbb{R}^{n_u}$, respectively. The
	constraints are formulated as chance constraints, i.e.\ by
	enforcing
\begin{align*}
	\Pr(x(k) \in \mathcal{X}) &\geq p_x \eqc \\
	\Pr(u(k) \in \mathcal{U}) &\geq p_u \eqc
\end{align*}
where $p_x$, $p_u$ are the associated satisfaction probabilities.
%
%===============================================================================
\subsection{Gaussian Process Regression}\label{ssc:gaussian_processes}
%===============================================================================
%
Gaussian process regression is a nonparametric framework for
nonlinear regression. A GP is a probability distribution over functions,
such that every finite sample of function values is jointly Gaussian
distributed. We apply Gaussian processes regression to infer the noisy
vector-valued function $g(x,u)$ in~\eqref{eq:system} from previously collected
measurement data of states and inputs $\{ (x_j , u_j), \, j = 0, \dots, \nt\}$.
State-input pairs form the input data to the GP and the corresponding
outputs are obtained from the deviation to the nominal system model:
\[
	y_j = g(x_j, u_j) + w_j =  B_d^\dagger \left(x_{j+1} - f(x_j, u_j) \right)  \eqc
\]
where $B_d^\dagger$ is the Moore-Penrose pseudo-inverse. Note that the
measurement noise on the data points $w_j$ corresponds to the process noise
in~\eqref{eq:system}. With $z_j:= [x^\tp_j,u^\tp_j]^\tp$, the data set of the GP is therefore
given by
\begin{align*}
	\mathcal{D} = \{ &\mathbf{y} = {[y_0, \ldots, y_M]}^\tp \in \mathbb{R}^{M\times n_d}, \\ &\mathbf{z} = {[z_0, \ldots, z_M]}^\tp  \in \mathbb{R}^{M\times n_z}\} \eqd
\end{align*}
Each output dimension is learned individually, meaning that we assume the
components of each $y_j$ to be independent, given the input data $z_j$.
Specifying a GP prior on $g$ in each output dimension $a \in \{1, \ldots, \nd
\}$ with kernel $k^a(\cdot,\cdot)$ and prior mean function $m^a(\cdot)$ results
in normally distributed measurement data with
\begin{equation}\label{eq:jointDataDistr}
	{[\mathbf{y}]}_{\cdot,a} \sim \mathcal{N}(m(\mathbf{z}),K_{\mathbf{z}\mathbf{z}}^a + I\sigma_a^2) \eqc
\end{equation}
where $K_{\mathbf{z}\mathbf{z}}^a$ is the Gram matrix of the data points, i.e.
${[K_{\mathbf{z}\mathbf{z}}^a]}_{ij} = k^a(z_i,z_j)$ and $m^a(\mathbf{z}) =
{[m^a(z_0),\ldots, m^a(z_\nt)]}^\tp$. 
The choice of kernel functions $k^a$ and its
parameterization is the determining factor for the inferred distribution of $g$
and is typically specified using prior process knowledge and
optimization~\cite{Rasmussen2006}, e.g. by optimizing the likelihood of the
observed data distribution~\eqref{eq:jointDataDistr}. Throughout this paper we
consider the squared exponential kernel function
\begin{equation}\label{eq:SE_kernel}
	k^a(z_i,z_j) = \sigma_{f,a}^2 \exp\!\left(-\frac{1}{2}{(z_i-z_j)}^\tp L_a^{-1} (z_i-z_j)\right) \eqc
\end{equation}
in which $L_a$ is a positive diagonal length scale matrix and $\sigma_{f,a}^2$ the
signal variance. It is, however,
straightforward to use any other (differentiable) kernel function.

The joint distribution of the training data and an arbitrary test point $z$ in
output dimension $a$ is given by
\begin{equation} \label{eq:jointDist}
	p({[y]}_a, {[\mathbf{y}]}_{\cdot,a}) =
	\mathcal{N}\left(\begin{bmatrix} m^a(\mathbf{z}) \\ m^a(z) \end{bmatrix},
	\begin{bmatrix} K^a_{\mathbf{z}\mathbf{z}} + I \sigma_a^2& K^a_{\mathbf{z}z}
			\\ K^a_{z\mathbf{z}} & K^a_{zz}\end{bmatrix} \right) \eqc
		\end{equation}
where ${[K^a_{\mathbf{z}z}]}_j = k^a(z_j,z)$, $K^a_{z\mathbf{z}} =
	{(K^a_{\mathbf{z}z})}^\tp$ and similarly $K^a_{zz} = k^a(z,z)$. The resulting
distribution of ${[y]}_a$ conditioned on the observed data points is again
Gaussian with $p({[y]}_a\,|\,{[\mathbf{y}]}_{\cdot,a}) =
	\mathcal{N}\left(\mu^d_a(z),\Sigma^d_a(z)\right)$ and
\begin{subequations}\label{eq:GPposterior}
	\begin{align}
		\mu_a^d(z)    & =  K_{z\mathbf{z}}^a{(K_{\mathbf{z}\mathbf{z}}^a + I \sigma^2_{a})}^{-1} {[\mathbf{y}]}_{\cdot,a} \eqc    \\
		\Sigma_a^d(z) & = K^a_{zz} - K_{z\mathbf{z}}^a{(K_{\mathbf{z}\mathbf{z}}^a + I \sigma^2_{a})}^{-1} K_{\mathbf{z}z}^a \eqd
	\end{align}
\end{subequations}
The resulting multivariate GP approximation of the unknown function $g(z)$ is then simply given
by stacking the individual output dimensions, i.e.
\begin{equation}\label{eq:multiGP}
	d(z) \sim \mathcal{N}\!\left(\mu^d(z),\Sigma^d(z)\right)
\end{equation}
with $\mu^d = [\mu^d_1,\ldots,\mu^d_\nd]^\tp$ and $\Sigma^d =
\diag([\Sigma^d_1,\ldots,\Sigma^d_\nd]^\tp)$. 

Evaluating mean and variance
in~\eqref{eq:GPposterior} has cost $\mathcal{O}(\nd \nz \nt)$ and
$\mathcal{O}(\nd \nz \nt^2)$, respectively and thus scales with the number of
data points. For large amounts of data points or fast real-time applications
this can
limit the use of GP models. To overcome these issues, various approximation
techniques have been proposed, one class of which is sparse Gaussian processes
using inducing inputs~\cite{Quinonero-Candela2007}, briefly outlined in the
following.
%
%
%===============================================================================
\subsection{Sparse Gaussian Processes}\label{ssc:sGP}
%===============================================================================
%
Many sparse GP approximations make use of \emph{inducing} targets
$\mathbf{y}_{\text{ind}}$, inputs $\mathbf{z}_{\text{ind}}$ and conditional distributions to
approximate the joint distribution~\eqref{eq:jointDist}~\cite{Quinonero-Candela2005}. 
Many such approximations exist, a popular of which is the Fully
Independent Training Conditional (FITC)~\cite{Snelson2006}, which we make use of
in this paper. Given a selection of inducing inputs $\mathbf{z}_{\text{ind}}$ and using
the shorthand notation $Q^a_{\zeta\tilde{\zeta}} :=
	\nobreak K^a_{\zeta\mathbf{z}_{\text{ind}}}
	{(K^a_{\mathbf{z}_{\text{ind}}\mathbf{z}_{\text{ind}}})}^{-1}
	K^a_{\mathbf{z}_{\text{ind}}\tilde{\zeta}}$ 
	the approximate posterior distribution
	is given by
\begin{subequations}\label{eq:sGP}
	\begin{align}
		\tilde{\mu}^{d}_a(z)    & = Q^a_{z\mathbf{z}}{(Q^a_{\mathbf{z}\mathbf{z}} + \Lambda)}^{-1}{[\mathbf{y}]}_{\cdot,a} \eqc \\
		\tilde{\Sigma}^{d}_a(z) & = K^a_{zz} - Q^a_{z\mathbf{z}}{(Q^a_{\mathbf{z}\mathbf{z}} + \Lambda)}^{-1}Q_{\mathbf{z}z}
	\end{align}
\end{subequations}
with $\Lambda = \diag(K^a_{\mathbf{z}\mathbf{z}} - Q^a_{\mathbf{z}\mathbf{z}} +
	I \sigma^2_{a})$. Concatenating the output dimensions similar
	to~\eqref{eq:GPposterior} we arrive at the approximation
\begin{equation*}
	\tilde{d}(z) \sim \mathcal{N}\!\left( \tilde{\mu}^d(z),\tilde{\Sigma}^d(z) \right) \eqd
\end{equation*}

Several of the matrices used in~\eqref{eq:sGP} can be precomputed and the
evaluation complexity becomes independent of the number of original data points.
With $\tilde{M}$ being the number of inducing points, this results in $\mathcal{O}( \nd
\nz \tilde{M} )$ and $\mathcal{O}( \nd \nz \tilde{M}^2 ) $ for the predictive
mean and variance, respectively.

There are numerous options for selecting the
inducing inputs, e.g.\ heuristically as a subset of the original data points, by
treating them as hyperparameters and optimizing their
location~\cite{Snelson2006}, or letting them coincide with test
points~\cite{Tresp2000}, which is often referred to as \emph{transductive}
learning. In Section~\ref{ssc:DynamicSparseGP} we make use of such transductive
ideas and propose a dynamic selection of inducing points, with a resulting local
approximation tailored to the predictive control task. 
% While these sparse approximations considerably speed up GP evaluation, some
% applications require additional steps to meet computational limits. 
% !TeX root = gpmpc.tex
%%%%%%%%%%%%%%%%%%%%%%%%%%%%%%%%%%%%%%%%%%%%%%%%%%%%%%%%%%%%%%%%%%%%%%%%%%%%%%%%
\section{MPC Controller Design}\label{sc:controller_design}
%%%%%%%%%%%%%%%%%%%%%%%%%%%%%%%%%%%%%%%%%%%%%%%%%%%%%%%%%%%%%%%%%%%%%%%%%%%%%%%%
%
We consider the design of an MPC controller for
system~\eqref{eq:system} using a GP approximation $d$ of the unknown function
$g$:
\begin{align}\label{eq:system_MPC}
	x_{i+1} = f(x_i,u_i) + B_d \left( d(x_i,u_i) + w_i \right)\eqd
\end{align}
At each time step, the GP approximation evaluates to a stochastic distribution
according to the residual model uncertainty and process noise, which is then
propagated forward in time. A stochastic MPC formulation allows the principled
treatment of chance constraints, i.e.\ by 
imposing a prescribed maximum probability of constraint violation.
Denoting a state and input reference trajectory by $X^r =
\{x^r_0,\ldots,x^r_N\}$, $U^r = \{u^r_0,\ldots,u^r_{N-1} \}$, respectively, the resulting stochastic finite time optimal control problem can be formulated
as
\begin{mini!}
{\Pi(x)}{\E\!\left(l_f(x_{N}\!-\!x^r_N) + \sum_{i=0}^{N-1} l(x_i\!-\!x^r_i,
	u_i\!- \!u^r_i)
	\right)\label{eq:costExpValue}} {\label{eq:org_optimization}}{}
\addConstraint{x_{i+1}}{= f(x_{i},u_{i}) + B_d (d(x_{i},u_{i}) + w_i)}
\addConstraint{u_i}{= \pi_i(x_i)}
\addConstraint{\Pr(x_{i+1}}{\in \mathcal{X}) \geq p_x\label{eq:jointChanceState}}
\addConstraint{\Pr(u_{i}}{\in \mathcal{U}) \geq p_u\label{eq:jointChanceInput}}
\addConstraint{x_{0}}{= x(k) \eqc}
\end{mini!}
for all $i=0,\ldots,N\!-\!1$ with appropriate terminal cost $l_f(x_{N})$ and
stage cost $l(x_{i}, u_{i})$, where the optimization is carried out over a
sequence of input policies $\Pi(x) = \{ \pi_0(x), \ldots, \pi_{N-1}(x) \}$.

In the form~\eqref{eq:org_optimization}, the optimization problem is
computationally intractable. In the following sections, we present techniques
for deriving an efficiently solvable approximation. Specifically, we use affine
feedback policies with precomputed linear gains and show how this allows simple
approximate propagation of the system uncertainties in terms of mean and
variance in prediction. Using these approximate distributions, we present a
framework for reformulating chance constraints deterministically and briefly
discuss the evaluation of possible cost functions.
%
%
%===============================================================================
\subsection{Ancillary Linear State Feedback Controller}\label{ssc:ancillary_cntr}
%===============================================================================
%
Optimization over general feedback policies $\Pi(x)$ is an infinite dimensional
optimization problem. A common approach that
integrates well in the presented framework is to restrict the policy class to
linear state feedback controllers
\[
	\pi_i(x_{i}) = K_i (\mu^x_i - x_{i}) + \mu^u_{i} \eqc
\]
where $\mu^x_i$ is the predicted mean of the state distribution. Using
pre-selected gains $K_i$, we then optimize over $\mu^u_{i}$, the mean of the
applied input. Note that due to the ancillary control law the input $u_i$
applied in prediction is a random variable, resulting from the affine
transformation of $x_i$.

\begin{figure}
	\centering
	\setlength\figureheight{2.5cm}
	\setlength\figurewidth{7.5cm}
	\input{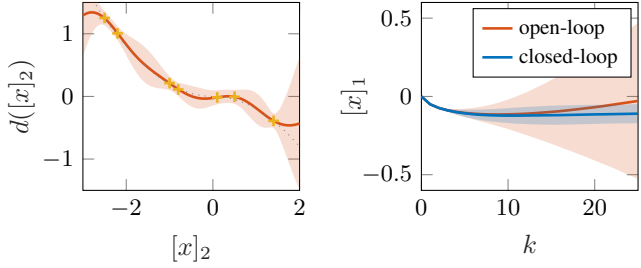}
	\caption{Propagation of uncertainty for double integrator with GP trained on
		a nonlinear friction term. The left plot displays the GP with
		$2$-$\sigma$ confidence bound, while the right plot shows the mean and
		$2$-$\sigma$ variance of repeated simulation runs of the system under an
		open-loop control sequence and closed-loop state feedback control.}\label{fg:ex_state_feedback}
\end{figure}

In Fig.~\ref{fg:ex_state_feedback}, the effect of state feedback on the
propagation of uncertainty is exemplified. It shows the evolution over time of a
double integrator with a quadratic friction term inferred by a GP\@. The right
plot displays mean and \mbox{2-$\sigma$} variance of a number of trajectories
from repeated simulations of the system with different noise realizations. Mean
and variance under an open-loop control sequence derived from a linear quadratic
infinite time optimal control problem are shown in red. The results with a
corresponding LQR feedback law applied in closed-loop are shown in blue. As
evident, open-loop input sequences can lead to rapid growth of uncertainty in
the prediction, and thereby cause conservative control actions in the presence
of chance constraints. Linear ancillary state feedback controllers are therefore
commonly employed in stochastic and robust MPC~\cite{Bemporad1999a}.

The adequate choice of ancillary feedback gains $K_i$ is generally a hard
problem, especially if the system dynamics are highly nonlinear. A useful
heuristic is to consider a linearization of the system dynamics around an
approximate prediction trajectory, which in MPC applications is typically
available using the solution trajectory of the previous time step. Feedback
gains for the linearized system can be derived e.g.\ by solving a finite horizon
LQR problem~\cite{Rawlings2009}. For mild or stabilizing nonlinearities, a fixed
controller gain $K_i = K$ can be chosen to reduce computational burden, as e.g.\
done for the example in Fig.~\ref{fg:ex_state_feedback}. 
%
%
%======================================================================
\subsection{Uncertainty Propagation}\label{ssc:unc_propagation}
%======================================================================
\subsubsection{Approximation as Normal Distributions}
%======================================================================
Because of stochastic process noise and the representation by a GP model, future
predicted states result in stochastic distributions. Evaluating the posterior
of a GP from an input distribution is generally intractable and the resulting
distribution is not Gaussian~\cite{Quinonero-Candela2003}. While under certain
assumptions on $g$, some strict over-approximations of the resulting
distributions exist~\cite{Koller2018}, they are typically very conservative and
computationally demanding. We focus instead on computationally cheap and
practical approximations at the cost of strict guarantees. 

State, control input and nonlinear disturbance are approximated as jointly
Gaussian distributed at every time step.
\begin{align}
	\begin{bmatrix}
		x^\tp_{i} & u^\tp_{i} & {(d_{i} + w_i)}^\tp
	\end{bmatrix}^\tp \sim \mathcal{N}\left( \mu_i, \Sigma_i \right) \nonumber \\ 
	=	\mathcal{N} \! \left(
	\begin{bmatrix}
			\mu^x_i \\ \mu^u_i \\ \mu^d_i
		\end{bmatrix} ,
	\begin{bmatrix}
			\Sigma^{x}_i & \Sigma^{xu}_i & \Sigma^{xd}_i \\ \star & \Sigma^{u}_i & \Sigma^{ud}_i \\ \star & \star &  \Sigma^d_i + \Sigma^w
		\end{bmatrix}
	\right) \eqc \label{eq:GaussianApprox}
\end{align}
where $\Sigma^u_i = K_i \Sigma^x_i K_i^\tp$, $\Sigma^{xu}_i = \Sigma^x_i K_i^\tp$
and $\star$ denotes terms given by symmetry. 
Considering the covariances between states, inputs and GP, i.e.\ $\Sigma^{xd}$
and $\Sigma^{ud}$, is of great importance for an accurate propagation of
uncertainty when the GP model $d$ is paired with a nominal system model $f$.
Using a linearization of the nominal system dynamics around the mean
\[ 
	f(x,u) \approx f(\mu^x,\mu^u) + \nabla \!  f(\mu^x,\mu^u) 
	\left(\begin{bmatrix} x \\ u\end{bmatrix} -
	\begin{bmatrix} \mu^x \\ \mu^u \end{bmatrix}\right) \eqc 
\]
similar to extended Kalman filtering, this permits simple update equations
for the state mean and variance based on affine transformations of the Gaussian
distribution in~\eqref{eq:GaussianApprox}
\begin{subequations}
	\begin{align*}
		\mu^{x}_{i+1}  & = f(\mu^x_i, \mu^u_{i}) + B_d \mu^d_i 
		\eqc                                                                     \\
		\Sigma^x_{i+1} & = \left[\nabla \!  f(\mu_i^x,\mu_i^u) \ B_d\right] \Sigma_i
		{\left[\nabla \!  f(\mu_i^x,\mu_i^u)\ B_d \right]}^\tp\eqd
	\end{align*}
\end{subequations}
%
%===============================================================================
\subsubsection{Gaussian Process Prediction from Uncertain Inputs}\label{sc:approximation}
%===============================================================================
%
In order to define $\mu^d_i, \Sigma^d_i$, $\Sigma^{xd}_i$ and $\Sigma^{ud}_i$
different approximations of the posterior of a GP from a Gaussian input have
been proposed in the literature. In the following, we will give a brief overview
of the most commonly used techniques, for details please refer
to~\cite{Quinonero-Candela2003,Kuss2006,Deisenroth2010}. For comparison, we
furthermore state the computational complexity of these methods.
For notational convenience, we will use
\[
	\mu^z_i 		= \begin{bmatrix} \mu^x_i \\ \mu^u_i \end{bmatrix}, \ 
	\Sigma^z_i 		= \begin{bmatrix} 
		\Sigma^x_i & \Sigma^{xu}_i \\ 
		\star      & \Sigma^u_i	
	\end{bmatrix} , \ 
	\Sigma^{zd}_i 	= \begin{bmatrix} 
		\Sigma^{xd}_i \\
		\Sigma^{ud}_i 
	\end{bmatrix} \eqd
\]
\paragraph{Mean Equivalent Approximation}\label{ssc:mean_eq}
A straightforward and computationally cheap approach is to
evaluate~\eqref{eq:multiGP} at the mean, such that
\begin{subequations}\label{eq:mean_eq}
	\begin{align}
		\mu^d_i & = \mu^d(\mu^z_i) \eqc             \\
		\begin{bmatrix} \Sigma^{zd}_i \\ \Sigma^d_i \end{bmatrix}
		        & = \begin{bmatrix} 0 \\  \Sigma^d(\mu^z_i) \end{bmatrix} \eqd
	\end{align}
\end{subequations}
In~\cite{Girard2002} it was demonstrated that this can lead to poor
approximations for increasing prediction horizons as it neglects the
accumulation of uncertainty in the GP\@. The fact that this approach neglects
covariance between $d$ and $(x,u)$ can in addition severely deteriorate the
prediction quality when paired with a nominal system $f(x,u)$.
The computationally most expensive operation is the matrix multiplication
in~\eqref{eq:GPposterior} for each output dimension of the GP, such that the
complexity of one prediction step is $\mathcal{O}(n_d n_z \nt^2)$.
%
%
%*******************************************************************************
\paragraph{Taylor Approximation}\label{ssc:taylor_approx}
%*******************************************************************************
%
Using a first-order Taylor approximation of~\eqref{eq:GPposterior}, the expected
value, variance and covariance of the resulting distribution results in
\begin{subequations}\label{eq:taylor_approx}
	\begin{align}
		\mu^d_i                    & = \mu^d(\mu^z_i) \eqc \\
		\begin{bmatrix}
			\Sigma^{zd}_i \\ \Sigma^d_i 
		\end{bmatrix} & = 
		\begin{bmatrix}
			\Sigma^{z}_i {(\nabla \!  \mu^d(\mu^z_i))}^\tp \\
			\Sigma^d(\mu^z_i) + \nabla \!  \mu^d(\mu^z_i) \Sigma^{z}_i {(\nabla \!  \mu^d(\mu^z_i))}^\tp
		\end{bmatrix} \, .
	\end{align}
\end{subequations}
Compared to~\eqref{eq:mean_eq} this leads to correction terms for the posterior
variance and covariance, taking into account the gradient of the posterior mean
and the variance of the input of the GP\@.
The complexity of one prediction step amounts to $\mathcal{O}(n_d n^2_z \nt^2)$.
Higher order approximations are similarly possible at the expense of increased
computational effort. 
%
%	\textcolor{red}{Computational complexity: \begin{enumerate} \item
%	    kernel/mean derivative at new datapoint $(z,z^\tp)$: \[\mathcal{O}(\nt
%	    \mathcal{O}_{dk}(n_z)) \approx {\mathcal{O}(\nt n_z}) \] \item Matrix
%	    multiplication: \[\nabla \!  k_a K^{-1} = [n_z \times \nt] [\nt \times \nt]
%	    = \mathcal{O}(n_z \nt^2)\] \item Matrix multiplication:  \[\nabla \! 
%	    \mu^d(\mu^z_i) \Sigma^z_i (\nabla \!  \mu^d(\mu^z_i))^\tp = [1 \times n_z]
%	    [n_z \times n_z] [n_z \times 1] = \mathcal{O}(n_z^2)\] \item For all
%	    dimension $n_d$: \end{enumerate} Computational complexity $\rightarrow
%	    \mathcal{O}(n_d (n_z \nt^2 + n_z^2)) \approx \mathcal{O}(n_d  n_z \nt^2
%	    )$}
%
%
%*******************************************************************************
\paragraph{Exact Moment Matching}\label{ssc:exact_mm}
%*******************************************************************************
%
It has been shown that for certain mean and kernel functions, the first and
second moments of the posterior distribution can be analytically
computed~\cite{Quinonero-Candela2003}. Under the assumption that $z_i = {[x_i^\tp, u_i^\tp]}^\tp$ is normally distributed, we can define
\begin{subequations}\label{eq:exact_mm}
	\begin{align}
		\mu^d_i                    & = \mathbb{E}\left(d(z_i)\right) \eqc \\
		\begin{bmatrix}
			\Sigma^{zd}_i \\
			\Sigma^d_i
		\end{bmatrix} & = 
		\begin{bmatrix}
			\text{cov}\left(z_i,d(z_i)\right) \\
			\text{var}\left(d(z_i)\right) \, .
		\end{bmatrix}
	\end{align}
\end{subequations}
In particular, this is possible for a zero prior mean function and the squared
exponential kernel. As the procedure exactly matches first and second moments of
the posterior distribution, it can be interpreted as an optimal fit of a
Gaussian to the true distribution. Note that the model formulation directly
allows for the inclusion of linear prior mean functions, as they can be
expressed as a nominal linear system in~\eqref{eq:system_MPC}, while preserving
exact computations of mean and variance. The computational
complexity of this approach is $\mathcal{O}(n_d^2 n_z
	\nt^2)$~\cite{Deisenroth2010}.
\begin{figure}
	\centering
	\setlength\figureheight{3cm}
	\setlength\figurewidth{7cm}
	\input{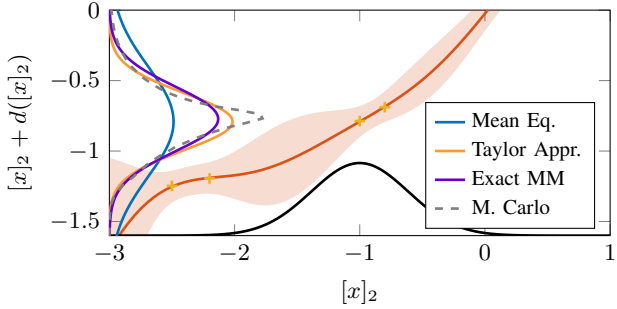}
	\caption{Comparison of prediction methods for Gaussian input ${[x]}_2$. The
		posterior distribution ${[x]}_2 + d({[x]}_2)$ is evaluated with the
		different approximation methods of Section~\ref{sc:approximation}. For
		reference, the true distribution is approximated by Monte Carlo
		simulation.}\label{fg:pred_comparison2}
\end{figure}
In Fig.~\ref{fg:pred_comparison2} the different approximation methods are
compared for a one-step prediction of the GP shown in
Fig.~\ref{fg:ex_state_feedback} with normally distributed input. As evident, the
predictions with \emph{Taylor Approximation} and \emph{Moment Matching} are
qualitatively similar, which will typically be the case if the second derivative
of the posterior mean and variance of the GP is small over the input
distribution. It is furthermore apparent that the posterior
distribution is not Gaussian and that the approximation as a Gaussian
distribution leads to prediction error, even if the first two moments are
matched exactly. This can lead to the effect that locally the \emph{Taylor
Approximation} provides a closer fit to the underlying distribution. The
\emph{Mean Equivalent Approximation} is very conservative by neglecting the
covariance between ${[x]}_2$ and $d({[x]}_2)$. Note that depending on the sign
of the covariance it can also be overly confident.

All presented approximations scale directly with the input and output
dimensions, as well as the number of data points and thus become expensive to
evaluate in high dimensional spaces. This presents a challenge for predictive
control approaches with GP models, which in the past have mainly focused on
relatively small and slow systems~\cite{Kocijan2004,Yang2015}. Using a nominal
model, however, it is possible to consider GPs that depend on only a subset of
states and inputs, such that the computational burden can be significantly
reduced. This is due to a reduction in the effective input dimension $n_z$, and
more importantly due to a reduction the in necessary training points $\nt$ for
learning in a lower dimensional space.
\begin{remark}
	With only slight notational changes, the presented approximation methods
	similarly apply to prior mean and kernel functions that are functions of
	only a subset of states and inputs.\label{prop:input_subset}
\end{remark}
Another significant reduction in the computational complexity can be achieved by
employing sparse GP approaches, for which
in the following we present a modification tailored to predictive control.
%
%===============================================================================
\subsection{Dynamic Sparse GPs for MPC}\label{ssc:DynamicSparseGP}
%===============================================================================
%
\begin{figure}[b!]
	\center{}
	\setlength\figureheight{6cm}
	\setlength\figurewidth{7.5cm}
	\input{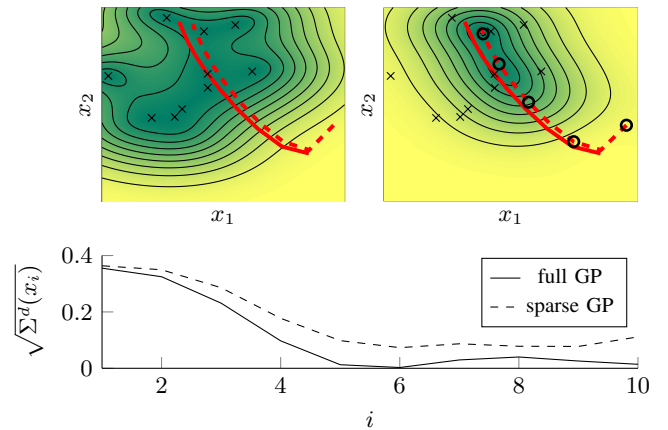}
	\caption{Illustration of dynamic sparse approximation~\cite{Hewing2018a}.
	Countor plot of the posterior variance of the full GP
	(top left) and dynamic sparse approximation (top right) with corresponding
	data points as black crosses. Trajectories
	planned by an MPC are shown as solid red lines, while the dashed lines show
	the prediction of the previous time step and used in the approximation, with
	the chosen inducing points indicated by black circles. The bottom plot shows
	the respective variances along the planned trajectory.}\label{fg:DSGP_toyEx}
\end{figure}
To reduce the computational burden one can make use of sparse GP approximations
as outlined in Section~\ref{ssc:sGP}, which often comes with little
deterioration of the prediction quality. A property of the task of predictive
control is that it lends itself to \emph{transductive} approximation schemes,
meaning that we adjust the approximation based on the test points, i.e. the
(future) evaluation locations in order to achieve a high fidelity local
approximation. In MPC, some knowledge of the test points is typically available
in terms of an approximate trajectory through the state-action space. This
trajectory can, e.g., be given by the reference signal or a previous solution
trajectory which will typically lie close to the current one. We therefore
propose to select inducing inputs locally at each sampling time according to the
prediction quality on these approximate trajectories. Ideally, the local
inducing inputs would be optimized to maximize predictive quality, which is,
however, not computationally feasible in the targeted millisecond sampling
times. Instead, inducing inputs are selected heuristically along the approximate
state input trajectory, specifically we focus here on the MPC solution computed
at a previous time step.

To illustrate the procedure, we consider a simple double integrator system
controlled by an MPC. Fig.~\ref{fg:DSGP_toyEx} shows 
the variance $\Sigma^d(x)$ of a GP trained on the systems' states. Additionally,
two successive trajectories from an MPC algorithm are displayed, in which the solid red line is the
current prediction, while the dashed line is the prediction trajectory from the
previous iteration. The plot on the left displays the original GP, with data
points marked as crosses, whereas on the right we have the sparse approximation
resulting from placing inducing points along the previous solution trajectory,
indicated by circles. The figure illustrates how full GP and sparse
approximation match closely along the predicted
trajectory of the system, while approximation quality far away from the
trajectory deteriorates. Since current and previous trajectory are similar,
however, local information is sufficient for computation of the MPC controller.
%
%===============================================================================
\subsection{Chance Constraint Formulation}
%===============================================================================
%
The tractable Gaussian approximation of the state and input distribution over
the prediction horizon in~\eqref{eq:GaussianApprox} can be used to
approximate the chance constraints in~\eqref{eq:jointChanceState} and~\eqref{eq:jointChanceInput}.

We reformulate the chance constraints on state and input w.r.t.\ their means
$\mu^x$, $\mu^u$ using constraint tightening based on the respective errors
$e^x_i = \mu^x_i - x_i$ and $\nobreak{e^u_i = K_i (\mu^x_i - x_i)}$. For this we make use
of probabilistic reachable sets~\cite{Hewing2018b}, an extension of the concept
of reachable sets to stochastic systems, which is related to probabilistic set
invariance~\cite{Kofman2012,Hewing2018}.
\begin{definition}[Probabilistic $i$-step Reachable Set]% \hfill \\
	A set $\set{R}$ is said to be a probabilistic $i$-step
	reachable set ($i$-step PRS) of probability level $p$ if
	\[
		\Pr(e_i \in \set{R} \,|\, e_0 = 0) \geq p \eqd
	\]
\end{definition}

Given $i$-step PRS $\set{R}^x$ of probability level $p_x$ for the state error
$e^x_i$ and similarly $\set{R}^u$ for the input $e^u_i$, we can define tightened constraints on $\mu^x_i$ and $\mu^u_i$
as
\begin{subequations}\label{eq:constrTightening}
	\begin{align} 
		\mu^x_i \in \set{Z} & = \set{X} \ominus \set{R}^x \eqc   \\
		\mu^u_i \in \set{V} & = \set{U} \ominus \set{R}^u \eqc
	\end{align}
\end{subequations}
where $\ominus$ denotes the Pontryagin set difference. Satisfaction of the
tightened constraints~\eqref{eq:constrTightening} for the mean thereby implies satisfaction of the
original constraints~\eqref{eq:jointChanceState}
and~\eqref{eq:jointChanceInput}, i.e. when $\mu^x_i \in \set{Z}$ we have $\Pr(x_i = \mu^x_i + e^x_i \in \set{X})
\geq \Pr(e^x_i \in \set{R}^x) \geq p_x$.

Under the approximation of a normal distribution of $x_i$, the
uncertainty in each time step is fully specified by the variance matrix
$\Sigma^x_i$ and $\Sigma^u_i$. The sets can then be computed as functions of
these variances, i.e. $\set{R}^x(\Sigma^x_i)$ and
$\set{R}^u(\Sigma^x_i)$, for instance as ellipsoidal confidence regions. 
The online computation and respective tightening in~\eqref{eq:constrTightening},
however, is often computationally prohibitive for general convex constraint
sets.

In the following, we present important cases for which a computationally cheap
tightening is possible, such that it can be performed online. For brevity, we
concentrate on state constraints, as input constraints can be treated
analogously. 
%
%*******************************************************************************
\subsubsection{Half-space constraints}\label{sssc:halfspace_tightening}
%*******************************************************************************
Consider the constraint set $\set{X}$ given by a single
half-space constraint $\nobreak{\set{X}^{hs} := \setst{x}{h^\tp x \leq b}}$, $h \in
	\mathbb{R}^n$, $b \in \mathbb{R}_+$. Considering the marginal distribution of
the error in the direction of the half-space $h^\tp e^x_i \sim \mathcal{N}(0, h^\tp
\Sigma_i^x h)$ enables us to use the quantile function of a
standard Gaussian random variable $\phi^{-1}(p_x)$ at the needed probability 
of constraint satisfaction $p_x$ to see that 
\[ \set{R}^x(\Sigma^x_i) := \setst{e}{h^\tp e \leq \phi^{-1}(p_x)\sqrt{h^\tp \Sigma^x_i h}} \]
is an $i$-step PRS of probability level $p_x$. In this case, evaluating the
Pontryagin difference in~\eqref{eq:constrTightening} is straightforward and we
can directly define the tightened constraint on the state mean
\[ \set{Z}^{hs}(\Sigma^x_i) := \setst{z}{h^\tp z \leq b -
		\phi^{-1}(p_x)\sqrt{h^\tp \Sigma^x_i h}} \eqd \]
\begin{remark}\label{rm:slab}
	A tightening for slab constraints $\nobreak{\set{X}^{sl} = \setst{x}{|h^\tp x| \leq b}}$
	can be similarly derived as
	\[ 
		\set{Z}^{sl}(\Sigma^x_i) := \setst{z}{|h^\tp z| \leq b - \phi^{-1}\!\left( \frac{p_x\!+\!1}{2}\right) \sqrt{h^\tp \Sigma^x_i h}} \eqd
	\]
\end{remark}
%
%*******************************************************************************
\subsubsection{Polytopic Constraints}
%*******************************************************************************
Consider polytopic constraints, given as the intersection of $n_j$
half-spaces
\[ 
	\set{X}^{p} = \setst{x}{h^\tp_j x \leq b_j\ \forall j = 1, \ldots, n_j} \eqd 
\]
Making use of ellipsoidal PRS for Gaussian random variables one
can formulate a semidefinite program to tighten
constraints~\cite{Zhou2013,Schildbach2015}, which is computationally demanding
to solve online. Computationally cheaper approaches typically rely on more
conservative bounds. One such bound is given by Boole's inequality
\begin{align}\label{eq:boolesIneq}
	 & \Pr(E_1 \wedge E_2) \leq \Pr(E_1) + \Pr(E_2) \eqc
\end{align}
which allows for the definition of polytopic PRS based on individual
half-spaces. We present two possibilities, the first of which is based on
bounding the violation probability of the individual polytope faces, and a
second which considers the marginal distributions of $e^x_i$.

\paragraph{PRS based on polytope faces}
Similar to the treatment of half space constraints, we can define a PRS on the
state error $e^x$ which is aligned with the considered constraints. Using
Boole's inequality~\eqref{eq:boolesIneq}, we have that
\begin{align*}
	&\set{R}^x(\Sigma^x_i) = \\ &\setst{e}{h_j^\tp e \leq \phi^{-1}\!\left(1\!-\!\frac{1\!-\!p_x}{n_j}\right)\sqrt{h_j^\tp \Sigma^x_i h_j} \ \forall j = 1, \ldots, n_j}
\end{align*}
is a PRS of probability level $p_x$. The tightening results in 
\begin{align*}
	&\set{Z}^p(\Sigma^x_i) = \\ &\setst{z}{h_j^\tp z \leq b_j - \phi^{-1}\!\left(\!1\!-\!\frac{1\!-\!p_x}{n_j} \!\right)\!\sqrt{h_j^\tp \Sigma^x_i h_j} \ \forall j = 1, \ldots, n_j}
\end{align*}
which scales with the number of faces of the polytope and can therefore
be quite conservative if the polytope has many (similar) faces. 

\paragraph{PRS based on marginal distributions} 
Alternatively, one can define a PRS based on the marginal distribution of the
error $e^x$ in each dimension, which therefore scales with the state dimension.
We use Boole's inequality~\eqref{eq:boolesIneq} with Remark~\ref{rm:slab}
directly on marginal distributions in each dimension of the state to define a
box-shaped PRS of probability level $p_x$ 
\begin{align*}
	\set{R}^x(\Sigma^x_i) = \setst{e}{ |[e]_j| \leq \phi^{-1}(\bar{p})\sqrt{[\Sigma^x_i]_{j,j}} \ \forall j = 1, \ldots, n_x} \eqd
\end{align*}
with $\bar{p} = 1-\left(\!\frac{1}{n_x}\!-\!\frac{p_x\!+\!1}{2n_x}\right)$.
To compute the Pontryagin difference~\eqref{eq:constrTightening} we make use of
the following Lemma.
\begin{lemma} Let $\set{A} =
	\setst{x}{H x \leq b}$ be a polytope and $\nobreak{\set{B} = \setst{e}{-r \leq e
	\leq r, \, r \in \mathbb{R}^n_+}}$,  a box in $\mathbb{R}^n$. 
	
	Then $\set{A} \ominus
	\set{B} = \setst{x}{H x \leq b - |H| r}$, where $|H|$ is the element-wise absolute value.
\end{lemma}
\begin{proof} From definition of the Pontryagin difference we have
\begin{align*}
	\set{A} \ominus \set{B} & = \setst{x}{x + e \in \set{A}\ \forall e \in \set{B}} \\
							& = \setst{x}{H(x + e) \leq b\ \forall e \in \set{B}} \\
							& = \setst{x}{[H]_{j,\cdot} x \leq [b]_j - \max_{e \in \set{B}} [H]_{j,\cdot} e \ \forall j = 1,\ldots, n }  \\
							& = \setst{x}{[H]_{j,\cdot} x \leq [b]_j - |[H]_{j,\cdot}| r \, \forall j = 1,\ldots, n } \\
							& = \setst{x}{H x \leq b - | H| r } \, ,
						\end{align*}
proving the result.
\end{proof}
The resulting tightened set is therefore
\[
	\set{Z}^p(\Sigma^x_i) = \setst{z}{H z \leq \tilde{b}(\Sigma^x_i)}
\]
with $\tilde{b}(\Sigma^x_i) = b - |H|\phi^{-1}(\bar{p})\sqrt{\diag(\Sigma^x_i)}$, where $\diag(\cdot)$ is the vector of
diagonal elements and the square root of the vector is taken element wise.

% \todo{The next two I would remove since they are combined in Lemma 1}
% \begin{lemma}[$\set{A} \ominus \set{B}$ for halfspace and box] Let $\set{A} =
% 		\setst{x}{h^\tp x \leq b}$ be a half-space and $\set{B} = \setst{e}{-r \leq e
% 			\leq r}$ a box in $\mathbb{R}^n$. Then $\set{A} \ominus \set{B} =
% 		\setst{x}{h^\tp x \leq b - |h|^\tp r}$.
% \end{lemma}
% \begin{proof}
% 	\begin{align*}
% 		\set{A} \ominus \set{B} & = \setst{x}{x + e \in \set{A}\ \forall e \in \set{B}} \\
% 		                        & = \setst{x}{h^\tp(x + e) \leq b\ \forall e \in \set{B}} \\
% 		                        & = \setst{x}{\max_{e \in \set{B}} h^\tp(x + e) \leq b}   \\
% 		                        & = \setst{x}{h^\tp x + |h|^\tp r \leq b}
% 	\end{align*}
% 	where $|h|$ is the element wise absolute value.
% \end{proof}
% \begin{lemma}
% 	$\bigcap_i \left(\set{A}_i\right) \ominus \set{B} = \bigcap_i
% 		\left(\set{A}_i \ominus \set{B} \right)$.
% \end{lemma}
% \begin{proof}
% 	Directly from definition? 
% 	\[
% 		\bigcap_i \left(\set{A}_i\right) = \setst{x}{x \in\set{A}_i \ \forall i}
% 	\]
% 	\[
% 		\bigcap_i \left(\set{A}_i\right) \ominus \set{B} = \setst{x}{x + e
% 			\in\set{A}_i \ \forall i \ \forall e \in \set{B}}
% 	\] and similarly
% 	\[
% 		\bigcap_i \left(\set{A}_i \ominus \set{B} \right) = \setst{x}{x + e
% 			\in\set{A}_i \ \forall e \in \set{B} \ \forall i}
% 	\]
% 	\todo{I think the swapping shouldn't make a difference}
% \end{proof}
\begin{remark}\label{rm:individual constraints}
	Treating polytopic constraints through~\eqref{eq:boolesIneq} can lead to
	undesired and conservative individual constraints~\cite{Lorenzen2017}.
	Practically, it can therefore be beneficial to constrain the probability of
	violating each individual half-space constraint.
\end{remark}
%
%===============================================================================
\subsection{Cost Function}\label{ssc:cost_function}
%===============================================================================
%
Given the approximate joint normal distribution of state and input,
the cost function~\eqref{eq:costExpValue} can be evaluated using standard
stochastic formulations. For simplicity, we focus here on costs in form of
expected values, which encompass most stochastic objectives typically
considered. While the expected value for general cost functions needs to be
computed numerically, there exist a number of functions for which evaluation
based on mean and variance information can be analytically expressed and is
computationally cheap. The most prominent example for tracking tasks is a
quadratic cost on states and inputs 
\begin{equation}
	l(x_i-x_i^r,u_i-u_i^r) = \Vert x_i - x_i^r \Vert_Q^2 + \Vert u_i - u_i^r \Vert_R^2 \label{eq:quad_cost}
\end{equation}
with appropriate weight matrices $Q$ and $R$, typically satisfying
$Q \succeq 0$ and $R \succ 0$, resulting in
\begin{align*}
	&\E \left(l(x_i-x_i^r,u_i-u_i^r)\right) = \\
	&\Vert \mu^x_i - x_i^r \Vert_Q^2 + \tr(Q \Sigma^x_i) + \Vert \mu^u_i - u_i^r \Vert_R^2 +\tr(R \Sigma^u_i) \eqd
\end{align*}
Further examples include a saturating cost~\cite{Deisenroth2010}, risk sensitive
costs~\cite{Whittle1981} or radial basis function networks~\cite{Kamthe2018}.

We refer to the evaluation of the expected value~\eqref{eq:costExpValue} in
terms of mean and variance as
\begin{align*}
	\E \left(l(x_i-x_i^r,u_i-u_i^r)\right) = c(\mu^x_i-x_i^r,\mu^u_i-u_i^r, \Sigma^x_i)
\end{align*} and similarly $c_f(\mu^x_N-x_N^r, \Sigma^x_N)$ for the terminal cost.
\begin{remark}
While many performance indices can be expressed as expected values, there exist
various stochastic cost measures, such as
conditional value-at-risk. Given the approximate state distributions, many can be
similarly considered.
\end{remark}

\subsection{Tractable MPC Formulation with GP Model}
%===============================================================================
%
By bringing together the approximations of the previous sections, the following
tractable approximation of the MPC problem~\eqref{eq:org_optimization} can be
derived:
\begin{mini}
{\{\mu^u_i\}}{ c_f(\mu^x_N\!-\!x^r_N, \Sigma^x_N) + \sum_{i=0}^{N-1}
c(\mu^x_i\!-\!x^r_i,	\mu^u_i\!-\!u^r_i, \Sigma^x_i)} {\label{eq:opt_final}}{}
\addConstraint{\mu^{x}_{i+1}}{= f(\mu^x_{i},\mu^u_{i}) + B_d \mu^d_i}
\addConstraint{\Sigma^x_{i+1}}{= \left[\nabla \!  f(\mu_i^x,\mu^u_i) \ B_d \right] \Sigma_{i}{\left[\nabla \!  f(\mu_i^x,\mu^u_i) \ B_d \right]}^\tp}
\addConstraint{\mu^{x}_{i+1}}{ \in \set{Z}(\Sigma^x_{i+1})}
\addConstraint{\mu^u_{i}}{ \in \set{V}(\Sigma^x_i)}
\addConstraint{\mu^{d}_{i}, \Sigma_i}{\text{ according to~\eqref{eq:GaussianApprox} and~\eqref{eq:mean_eq},~\eqref{eq:taylor_approx} or~\eqref{eq:exact_mm}}}
\addConstraint{\mu^x_0}{= x(k), \Sigma^x_0 = 0}
\end{mini}
for $i = 0,\ldots,N\!-\!1$. The resulting optimal control law is obtained in a
receding horizon fashion as $\kappa(x(k)) = \mu^{u*}_0$, where $\mu^{u*}_0$ is
the first element of the optimal control sequence $\{\mu^{u*}_0, \ldots,
\mu^{u*}_N\}$ obtained from solving~\eqref{eq:opt_final} at state $x(k)$.

The presented formulation of the MPC problem with a GP-based model results in a
non-convex optimization problem. Assuming twice differentiability of kernel and
prior mean function, second-order derivative information of all quantities is available. Problems of this form can typically be solved to local
optima using Sequential Quadratic Programming (SQP) or nonlinear interior-point
methods~\cite{Diehl2009}.
There exist a number of general-purpose solvers that can be applied to solve
this class of optimization problems, e.g. the openly available
IPOPT~\cite{Wachter2006}. In addition, there are specialized solvers that
exploit the structure of predictive control problems, such as
ACADO~\cite{Houska2011} or FORCES Pro~\cite{FORCESPro,Zanelli2017}, which
implements a nonlinear interior point method. Derivative information can be
automatically generated using automated differentiation tools, such as
CASADI~\cite{Andersson2013}. 

In the following, we demonstrate the proposed algorithm and its properties using
one simulation and one experimental example. The first is an illustrative
example of an autonomous underwater vehicle (AUV) which, around a trim point, is
well described by nominal linear dynamics but is subject to nonlinear friction
effects for larger deviations.
The second example is a hardware implementation demonstrating the approach for
autonomous racing of miniature race cars which are described by a nonlinear
nominal model. 
While model complexity and small sampling times require additional
approximations in this second example, we show that we can leverage key benefits
of the proposed approach in a highly demanding real-world application.

% !TeX root = gpmpc.tex

%%%%%%%%%%%%%%%%%%%%%%%%%%%%%%%%%%%%%%%%%%%%%%%%%%%%%%%%%%%%%%%%%%%%%%%%%%%%%%%%
\section{Online Learning for Autonomous Underwater Vehicle}\label{ssc:auv}
%%%%%%%%%%%%%%%%%%%%%%%%%%%%%%%%%%%%%%%%%%%%%%%%%%%%%%%%%%%%%%%%%%%%%%%%%%%%%%%%
%
We consider the depth control of an autonomous underwater vehicle (AUV). The
vehicle is described by a nonlinear continuous-time model of the vehicle at constant
surge velocity as ground truth and for simulation purposes~\cite{Naik2007}.
Around a trim point of purely horizontal movement, the states and input of the
system are 
\begin{description}
	\item[${[x]}_1$] the pitch angle relative to trim point [\si{rad}],
	\item[${[x]}_2$] the heave velocity relative to trim point [\si{m/s}],
	\item[${[x]}_3$] the pitch velocity [\si{rad/s}],
	\item[$u$] the stern rudder deflection around trim point [\si{rad}].
\end{description}
% and the model naturally decomposes into a linear and nonlinear part
% \begin{align}
% 	\dot{x}^c(t) = A^c x(t)  + B^c u^c(t) + B^c_d (g^c(x^c_1(t),x^c_2(t)) + w^c(t)) \eqc \label{eq:nonlinear_auv}
% \end{align}
% where $g^c(0) = 0$ and $\nabla \!  g^c(0) = 0$. The unknown function $g^c$ expresses
% nonlinear friction terms, which effect the velocity states through matrix $B^c_d = [ I, 0]$. We consider white process noise with power 
% spectral density $PSD = \text{diag}(4 \cdot 10^{-6},4 \cdot 10^{-2})$, which, 
% for simulation purposes, we discretize as $w^c(t) \sim \mathcal{N}(0,PSD/T_s)$,
% for $kT_s \leq t < (k+1)T_s$, where $T_s$ is the sampling time of the
% discrete-time model. Details of the simulation model are provided in Appendix
% \ref{ap:simulation_model}.

We assume that an approximate linear 
system model is given, which in practice can be
established using methods of linear system identification around the trim point
of the system. Using a zero-order hold discretization of the linear part of the
model with $T_s = \SI{100}{ms}$ we obtain
\begin{align*}
	x(k\!+\!1) & = A x(k) + B u(k) + B_d \left(g(x(k)) + w(k)\right) \eqc
\end{align*}
which is in the form of~\eqref{eq:system}. 
The nonlinearity results from friction effects when moving in the water, which is therefore modeled as only affecting the (continuous time) velocity states,
\[ 
	B_d = \begin{bmatrix} 0 & \frac{T_s^2}{2} \\ T_s & 0 \\ 0 & T_s \end{bmatrix}, \ g(x(k)) = g({[x(k)]}_2,{[x(k)]}_3) : \mathbb{R}^2 \rightarrow \mathbb{R}^2 \eqd
\]
Note that the eigenvalues of $A$ are given by $\lambda
= \{ 1.1, 1.03, 1, 0.727 \}$, i.e.\ the linear system has one integrating
and two unstable modes.

In the considered scenario, the goal is to track reference changes from the
original zero set point to a pitch angle of $30\degree$, back to zero and then
to $45\degree$. We furthermore consider a safety state constraint on the pitch
angle of at least $10\degree$ below the reference, as well as input constraints
corresponding to $\pm 20\degree$ rudder deflection. In this example, we treat
the case of \emph{online learning}, that is we start without any data
about the nonlinearity $g$, collect measurement data during operation and
enhance performance online.
%
%*******************************************************************************
\subsection{GP-based Reference Tracking Controller}\label{sssc:auv_controllerDesign}
%*******************************************************************************
%
GP data $\set{D} = \{\mathbf{y}, \mathbf{z} \}$ is generated by calculating the
deviation of the linear model from the measured states, as described in
Section~\ref{ssc:gaussian_processes}, where the input data is reduced to the
velocity states $z_j = [{[x_j]}_2, {[x_j]}_3]^\tp$. Data is continuously updated
during the closed-loop run. Specifically, we consider a new data point every 5
time steps, and keep track of 30 data points by discarding the oldest when
adding a new point. The training data is initialized with 30 data points of zero
input and zero output. We employ a squared exponential
kernel~\eqref{eq:SE_kernel} for both output dimensions with fixed
hyperparameters $L_1 = L_2 = \text{diag}([0.35,0.15]^\tp)$, and variances
$\sigma_{f,1}^2 = 0.04$ and $\sigma_{f,2}^2 = 0.25$.

We use a quadratic stage cost as in~\eqref{eq:quad_cost}, with weight matrices
$Q = \text{diag}([1, 0, 10, 0.5^\tp])$, $R = 20$ and a prediction horizon $N =
35$ to track a pitch angle reference. The terminal cost $P$ is chosen according
to the solution of the associated discrete-time algebraic Riccati equation, i.e.
the LQR cost. As ancillary linear controller $K_i$, an infinite horizon LQR
control law based on the linear nominal model is designed using the same weights
as in the MPC and used in all prediction steps $i$. This stabilizes the linear
system and reduces uncertainty growth over the prediction horizon. To propagate
the uncertainties associated with the GP we make use of the \emph{Taylor
Approximation} outlined in Section~\ref{sc:approximation}. Constraints are
introduced based on Remark~\ref{rm:individual constraints} considering a maximum
probability of individual constraint violation of $2.28\%$, corresponding to a
$2$-$\sigma$ confidence bound. The MPC optimization problem~\eqref{eq:opt_final}
is solved using FORCES Pro~\cite{FORCESPro, Zanelli2017}.
%
%===============================================================================
\subsection{Results}\label{sssc:auv_results}
%===============================================================================
%
\begin{figure}
	\center{}
	\setlength\figureheight{2.5cm}
	\setlength\figurewidth{7cm}
\definecolor{mycolor1}{rgb}{0.85000,0.32500,0.09800}%
%\definecolor{mycolor2}{rgb}{0.20392,0.90196,0.20392}%
\definecolor{mycolor2}{rgb}{ 0.1961,0.8039,0.1961}% 
\begin{tikzpicture}

\begin{axis}[%
width=0.951\figurewidth,
height=\figureheight,
at={(0\figurewidth,0\figureheight)},
scale only axis,
xmin=1.5,
xmax=6,
ymin=-0.4,
ymax=0.8,
axis background/.style={fill=white},
xlabel = {time [$\si{s}$]},
legend pos = south east,
]
\addplot[const plot, color=black, line width = 1pt] table[row sep=crcr] {%
0	0\\
4.5	0.555229202441568\\
6.09999999999999	0.555229202441568\\
};
\addlegendentry{$[x^r]_1$}

\addplot[const plot, color=mycolor1, dashed, forget plot, line width = 1pt] table[row sep=crcr] {%
0	-0.174532925199433\\
4.5	0.380696277242135\\
6.09999999999999	0.380696277242135\\
};

\addplot[const plot, color=mycolor2, dashed, forget plot, line width = 1pt] table[row sep=crcr] {%
0	-0.233124098676941\\
6.09999999999999	-0.233124098676941\\
};

\addplot[const plot, color=mycolor2, dashed, forget plot, line width = 1pt] table[row sep=crcr] {%
0	0.465007602120791\\
6.09999999999999	0.465007602120791\\
};

\addplot [color=mycolor1, line width=1.0pt, forget plot]
  table[row sep=crcr]{%
0	0\\
0.2225	0.000344900360326195\\
0.442297206319346	0.000498415706178923\\
0.615	0.000113706068555874\\
0.8625	-0.00106701944340659\\
1.0775	-0.00173887290087071\\
1.50028119412591	-0.00229400489884335\\
};

\addplot [color=mycolor1, dotted, line width=1.0pt]
  table[row sep=crcr]{%
1.5	-0.00229356456074559\\
1.6	-0.00095765107735879\\
1.7	0.00305482919195743\\
1.8	0.00918468059715227\\
1.9	0.0170013727487266\\
2	0.0261778744633796\\
2.1	0.0364692088992786\\
2.2	0.0476946289167053\\
2.3	0.0597233975669242\\
2.4	0.0724639044973809\\
2.5	0.0858556816137863\\
2.6	0.0998638214999517\\
2.7	0.11447528601164\\
2.8	0.129696563037599\\
2.9	0.14555205096244\\
3	0.162082409796753\\
3.1	0.179341937978363\\
3.2	0.197393887909124\\
3.3	0.216302661654241\\
3.4	0.236122199158329\\
3.5	0.256880685757332\\
3.6	0.278562811959518\\
3.7	0.301091732787595\\
3.8	0.324313011347402\\
4.1	0.395176332686886\\
4.2	0.417676959464317\\
4.3	0.438544686614382\\
4.4	0.456891496186883\\
4.5	0.471589259407684\\
};
\addlegendentry{$[x]_1$}

\addplot[area legend, draw=black, fill=mycolor1, draw opacity=0, fill opacity=0.3, forget plot]
table[row sep=crcr] {%
x	y\\
1.5	-0.00229356456074601\\
1.6	-0.000897452233385363\\
1.7	0.00385476199097761\\
1.8	0.0117225015908439\\
1.9	0.0221823969718487\\
2	0.0346963897044682\\
2.1	0.0488129580421021\\
2.2	0.0641743428447565\\
2.3	0.0805075001154634\\
2.4	0.0976112813739198\\
2.5	0.115344145648963\\
2.6	0.133613828685244\\
2.7	0.152369231623123\\
2.8	0.171594255875115\\
2.9	0.191302992946186\\
3	0.211535414134455\\
3.1	0.23235245475742\\
3.2	0.253829201696194\\
3.3	0.276044915901932\\
3.4	0.299069044920278\\
3.5	0.322943334282414\\
3.6	0.347661477832645\\
3.7	0.373148896564121\\
3.8	0.399245460600516\\
3.9	0.425692863882445\\
4	0.452126230177603\\
4.1	0.478067139700384\\
4.2	0.502912975598507\\
4.3	0.525914575352438\\
4.4	0.546130367769245\\
4.5	0.562346090201384\\
4.5	0.380832428613983\\
4.4	0.367652624604522\\
4.3	0.351174797876325\\
4.2	0.332440943330127\\
4.1	0.312285525673389\\
4	0.291381707787774\\
3.9	0.270271006752758\\
3.8	0.249380562094288\\
3.7	0.229034569011071\\
3.6	0.209464146086392\\
3.5	0.190818037232249\\
3.4	0.17317535339638\\
3.3	0.156560407406551\\
3.2	0.140958574122053\\
3.1	0.126331421199305\\
3	0.112629405459051\\
2.9	0.0998011089786945\\
2.8	0.0877988702000816\\
2.7	0.0765813404001559\\
2.6	0.0661138143146587\\
2.5	0.0563672175786107\\
2.4	0.0473165276208423\\
2.3	0.0389392950183856\\
2.2	0.0312149149886543\\
2.1	0.0241254597564545\\
2	0.0176593592222903\\
1.9	0.0118203485256037\\
1.8	0.00664685960346131\\
1.7	0.00225489639293767\\
1.6	-0.00101784992133186\\
1.5	-0.00229356456074601\\
}--cycle;

\addplot [color=mycolor2, line width=1.0pt, forget plot]
  table[row sep=crcr]{%
0	-2.60787771377835e-08\\
0.1	-2.60787771377835e-08\\
0.1	0.00112632211287722\\
0.2	0.00112632211287722\\
0.2	0.00351525051669399\\
0.3	0.00351525051669399\\
0.3	9.7096789001494e-05\\
0.4	9.7096789001494e-05\\
0.4	0.00187938587040204\\
0.5	0.00187938587040204\\
0.5	-0.00216879376434131\\
0.6	-0.00216879376434131\\
0.6	-0.00265938489306894\\
0.7	-0.00265938489306894\\
0.7	-0.00574570594558521\\
0.8	-0.00574570594558521\\
0.8	-0.00458641822974237\\
0.9	-0.00458641822974237\\
1	-0.00475254580979545\\
1	-0.00180226473162137\\
1.1	-0.00180226473162137\\
1.2	-0.00179306938360524\\
1.3	-0.00191233083591058\\
1.3	-0.00156908957005508\\
1.4	-0.00156908957005508\\
1.5	-0.00152342816974027\\
1.5	-0.102509348599269\\
};

\addplot [color=mycolor2, dotted, line width=1.0pt]
  table[row sep=crcr]{%
1.5	-0.102509348599269\\
1.6	-0.102509348599269\\
1.6	-0.0827711957751669\\
1.7	-0.0827711957751669\\
1.7	-0.0652956912561002\\
1.8	-0.0652956912561002\\
1.8	-0.0499918443640821\\
1.9	-0.0499918443640821\\
1.9	-0.0367688241539055\\
2	-0.0367688241539055\\
2	-0.0254766234718122\\
2.1	-0.0254766234718122\\
2.1	-0.0159270629866599\\
2.2	-0.0159270629866599\\
2.2	-0.00792415278604697\\
2.3	-0.00792415278604697\\
2.3	-0.0012851550668378\\
2.4	-0.0012851550668378\\
2.4	0.00414876623858618\\
2.5	0.00414876623858618\\
2.5	0.00851004039292036\\
2.6	0.00851004039292036\\
2.6	0.0119087703996028\\
2.7	0.0119087703996028\\
2.7	0.0144431544783146\\
2.8	0.0144431544783146\\
2.8	0.0162169149257458\\
2.9	0.0162169149257458\\
2.9	0.01736445297316\\
3	0.01736445297316\\
3	0.0180827012827063\\
3.1	0.0180827012827063\\
3.1	0.0186649321862502\\
3.2	0.0186649321862502\\
3.2	0.0195253998452172\\
3.3	0.0195253998452172\\
3.3	0.0211987102605748\\
3.4	0.0211987102605748\\
3.4	0.024297744313067\\
3.5	0.024297744313067\\
3.5	0.0294287334655428\\
3.6	0.0294287334655428\\
3.6	0.0370861405598601\\
3.7	0.0370861405598601\\
3.7	0.0475694190905553\\
3.8	0.0475694190905553\\
3.8	0.0609590572295255\\
3.9	0.0609590572295255\\
3.9	0.077164594538834\\
4	0.077164594538834\\
4	0.0960336343256705\\
4.1	0.0960336343256705\\
4.1	0.117522378569207\\
4.2	0.117522378569207\\
4.2	0.141976260010297\\
4.3	0.141976260010297\\
4.3	0.170556220771662\\
4.4	0.170556220771662\\
4.4	0.205582112211857\\
};
\addlegendentry{$u$}

\addplot[area legend, draw=black, fill=mycolor2, draw opacity=0, fill opacity=0.3, forget plot]
table[row sep=crcr] {%
x	y\\
1.5	-0.102509348599269\\
1.6	-0.0816168206935266\\
1.7	-0.0497126941035393\\
1.8	-0.0189362828150521\\
1.9	0.00922761873134837\\
2	0.0342808524988992\\
2.1	0.0562297833704087\\
2.2	0.0752915030293177\\
2.3	0.0917640102155233\\
2.4	0.105963858985628\\
2.5	0.11819685823406\\
2.6	0.128747930912204\\
2.7	0.137884151644181\\
2.8	0.145868316659326\\
2.9	0.152981396146107\\
3	0.159550860013837\\
3.1	0.16597845561464\\
3.2	0.172755303120595\\
3.3	0.180448667538517\\
3.4	0.18964720116757\\
3.5	0.200868961985489\\
3.6	0.214461427850803\\
3.7	0.230538432555531\\
3.8	0.248985997591787\\
3.9	0.269535981729126\\
4	0.291879994410986\\
4.1	0.315809830511625\\
4.2	0.341418258932491\\
4.3	0.369368381069181\\
4.4	0.400952948392228\\
4.4	0.0102112760314872\\
4.3	-0.0282559395258571\\
4.2	-0.0574657389118984\\
4.1	-0.08076507337321\\
4	-0.0998127257596458\\
3.9	-0.115206792651458\\
3.8	-0.127067883132736\\
3.7	-0.135399594374421\\
3.6	-0.140289146731083\\
3.5	-0.142011495054404\\
3.4	-0.141051712541436\\
3.3	-0.138051247017367\\
3.2	-0.133704503430161\\
3.1	-0.128648591242139\\
3	-0.123385457448424\\
2.9	-0.118252490199786\\
2.8	-0.113434486807834\\
2.7	-0.108997842687552\\
2.6	-0.104930390112999\\
2.5	-0.101176777448219\\
2.4	-0.0976663265084552\\
2.3	-0.0943343203491989\\
2.2	-0.0911398086014115\\
2.1	-0.0880839093437289\\
2	-0.0852340994425231\\
1.9	-0.0827652670391592\\
1.8	-0.0810474059131129\\
1.7	-0.0808786884086613\\
1.6	-0.0839255708568075\\
1.5	-0.102509348599269\\
}--cycle;

\end{axis}
\end{tikzpicture}%
	\caption{Predicted pitch angle ${[x]}_1$ and rudder deflection $u$. Dotted lines are the mean prediction and shaded the $2$-$\sigma$ confidence region. The dashed lines show the corresponding state and input constraints, while the black line is the pitch angle reference.}\label{fg:auv_pred}
\end{figure}
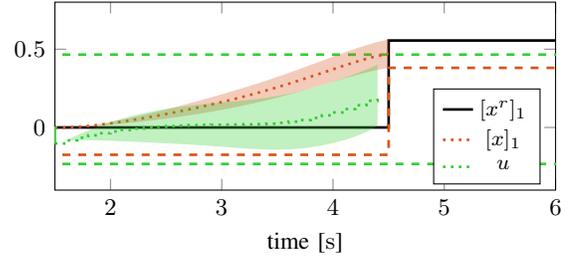
Fig.~\ref{fg:auv_pred} shows the prediction of the GP-based MPC for
the first reference change from $0\degree$ to $30\degree$. Since no data was
collected on the state
trajectory necessary for this change, the predicted state and input trajectories
are uncertain and a safety margin to the state constraint at the end of the
horizon is enforced. 
The resulting closed-loop trajectory, during which the system learns from online data,
is displayed in the top plot of
Fig.~\ref{fg:auv_comp}. For comparison, we run the same simulation with a soft
constrained linear MPC formulation with the same cost function, which does not
consider nonlinearities or uncertainties in the dynamics, shown in the
bottom plot. The results demonstrate improved performance of the GP-based MPC over the linear MPC, especially with regard to constraint satisfaction. The constraint on minimum pitch angle is violated under
the linear MPC control law during both reference changes, even 
though the soft constraint is chosen as an exact penalty function such that
constraints are always satisfied, if possible.
Fig.~\ref{fg:auv_err} exemplifies the residual model error during the closed-loop simulation applying 
the GP-based controller, as well as the predicted $2$-$\sigma$ residual error
bound of the GP. We observe the largest error
during reference changes, which are anticipated by the
GP-uncertainty, overall matching the resulting residual errors well.
\begin{figure}
	\begin{subfigure}{\linewidth}
		\begin{flushright}
			\setlength\figureheight{4cm}
			\setlength\figurewidth{7cm}
			\input{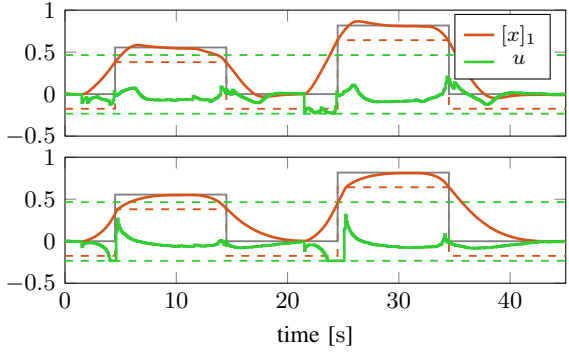} \hspace*{0.5cm} \\
			\caption{GP-based MPC (top) and linear MPC (bottom). The solid gray line shows the reference value of the pitch angle, dashed in the respective color the state and input constraints.}\label{fg:auv_comp}
		\end{flushright}
	\end{subfigure} 
	\begin{subfigure}{\linewidth}
		\vspace{0.2cm}
		\center{}
		\setlength\figureheight{2cm}
		\setlength\figurewidth{7cm}
		\input{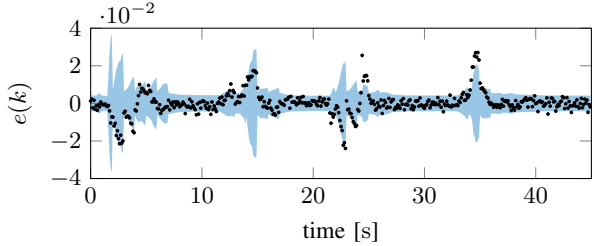} \hspace*{0.5cm} \\
		\caption{Residual model error with GP-based MPC. Measured model error 
		as black dots, 2-$\sigma$ residual model uncertainty from GP shaded blue.}\label{fg:auv_err}
	\end{subfigure}
	\caption{Simulation results of GP-based MPC for autonomous underwater vehicle in an online learning scenario.}
\end{figure}
%
%%%%%%%%%%%%%%%%%%%%%%%%%%%%%%%%%%%%%%%%%%%%%%%%%%%%%%%%%%%%%%%%%%%%%%%%%%%%%%%%
\section{Autonomous Racing}\label{sc:AutonomousRacing}
%%%%%%%%%%%%%%%%%%%%%%%%%%%%%%%%%%%%%%%%%%%%%%%%%%%%%%%%%%%%%%%%%%%%%%%%%%%%%%%%
%
%
As a second example we consider an autonomous racing scenario, in which the goal
is to drive a car around a track as quickly as possible, while keeping the
vehicle safe, i.e.\ while avoiding collision with the track boundaries. The
controller is based on a model predictive contouring control
formulation~\cite{Faulwasser2009, Lam2010} which has been applied to the problem
of autonomous racing in~\cite{Liniger2015}. Preliminary simulations results were
published in~\cite{Hewing2018a}.
%===============================================================================
\subsection{Car Dynamics}\label{ssc:NominalDynamics}
%===============================================================================
%
The race cars are modeled by continuous-time nominal dynamics $\dot{x} = f^c(x,u)$ obtained from a bicycle model with nonlinear tire forces given by a
simplified Pacejka tire model~\cite{Pacejka1992},
which results in the following states and inputs 
\[ x = [X,Y,\Phi,v_x,v_y,\omega]^\tp, \ u = [p, \delta]^\tp \eqc \] with position
$x^{XY} = {[X,Y]}^\tp$, orientation $\Phi$, longitudinal and lateral velocities $v_x$ and
$v_y$, and yaw rate $\omega$. The inputs to the system are the motor duty cycle
$p$ and the steering angle $\delta$. For details on the system modeling
please refer to~\cite{Liniger2015,Hewing2018a}. 

For use in the MPC formulation, we discretize the system using a Runge-Kutta
method with a sampling time of $T_s = \SI{20}{ms}$. In order to account for
model mismatch due to inaccurate parameter choices and limited fidelity of this
simple model, we add $g(x,u)$ capturing unmodeled dynamics, as well as additive
Gaussian white noise $w$. Due to the structure of the nominal model, i.e.\ since
the dynamics of the first three states are given purely by kinematic
relationships, we assume that the model uncertainty, as well as the process
noise $w$, only affect the velocity states $v_x$, $v_y$ and $\omega$ of the
system. From physical considerations, we can also infer that the unmodeled
dynamics do not depend on the position states, i.e.\ we assume 
\[  B_d = [0 \ I]^\tp , \ g(x,u) = g(v_x,v_y,\omega, p, \delta) : \mathbb{R}^5
	\rightarrow \mathbb{R}^3 \eqc \] resulting in
\begin{align*}
	x(k\!+\!1) = f(x(k),u(k)) \! + \! B_d (g(x(k),u(k)) \! + \! w(k)) \eqd 
\end{align*}

The system is subject to input constraints $\mathcal{U}$, i.e.\
the steering angle is limited to lie in $ \pm \delta_{\max}$ and the
duty cycle has to lie $[-0.1,1]$, where the negative values correspond to
negative applied torques, i.e. breaking. Additionally, operation requires the
vehicle to stay on track, which is expressed as a state constraint $\mathcal{X}$
on the car's position.
\begin{figure}[b]
	\center{}
	\includegraphics[width = 0.8\linewidth]{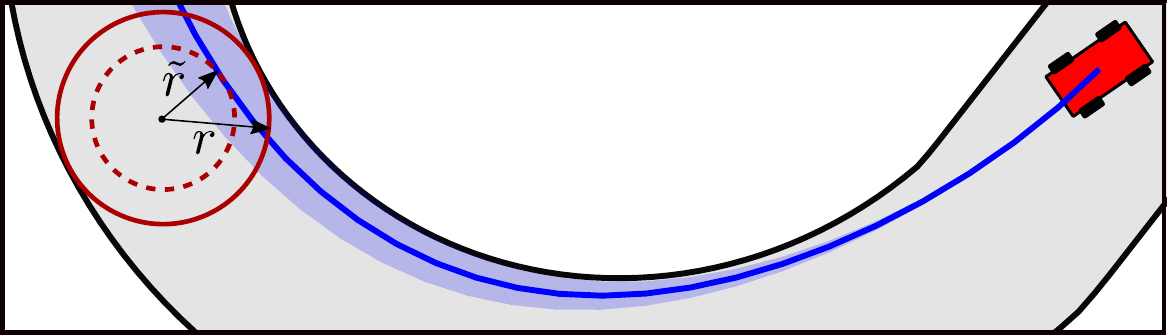}
	\caption{Illustration of the constraint tightening procedure. The effective track radius is adjusted based on the predicted position uncertainty.}\label{fg:raceCar_constrTightening}
\end{figure}
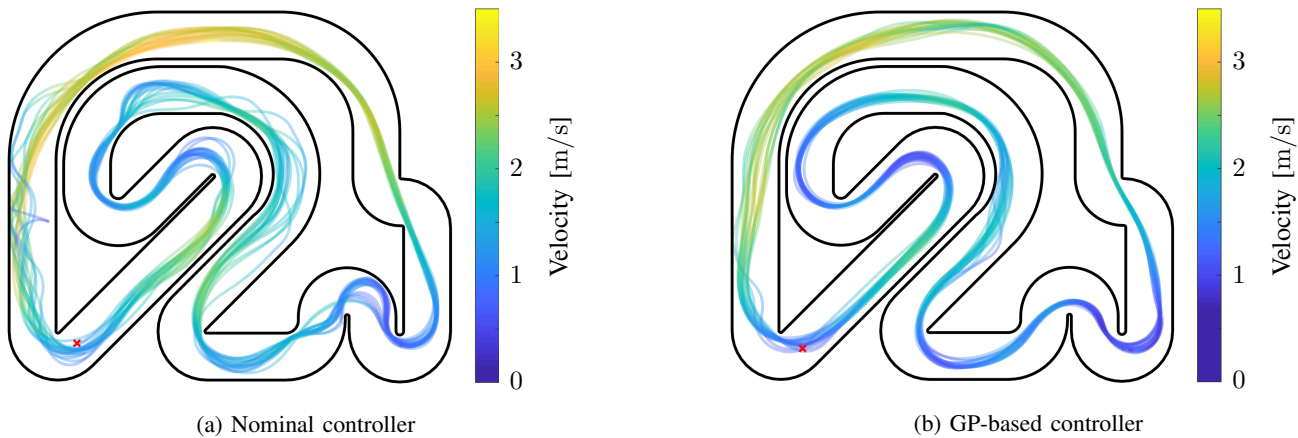
\begin{figure*}[h!]
	\begin{subfigure}{0.45\textwidth}
		\begin{tikzpicture}
			\node (pic) {\includegraphics[width =
			0.8\textwidth]{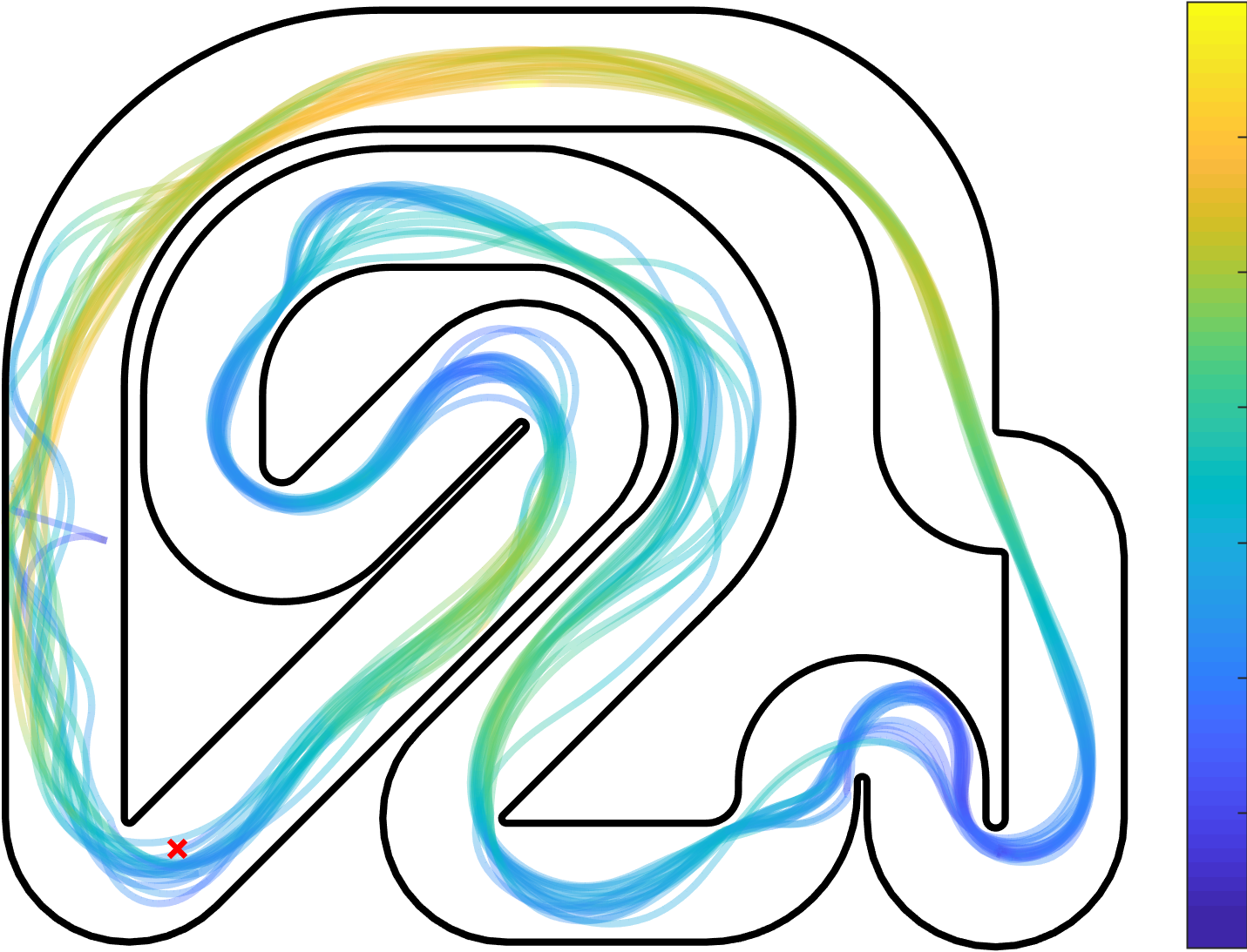}};
			\node[right =0.7cm of pic, rotate=90, xshift = -1.2cm] {Velocity [$\si{m/s}$]};
			\node[right =-0.1cm of pic, yshift = -2.45cm] {$0$};
			\node[right =-0.1cm of pic, yshift = -1.05cm] {$1$};
			\node[right =-0.1cm of pic, yshift = 0.38cm] {$2$};
			\node[right =-0.1cm of pic, yshift = 1.8cm] {$3$};
		\end{tikzpicture}
		\caption{Nominal controller}\label{fg:ORCA_nominal}
	\end{subfigure}
	\hspace{1.2cm}
	\begin{subfigure}{0.45\textwidth}
		\begin{tikzpicture}
			\node (pic) {\includegraphics[width =
			0.8\textwidth]{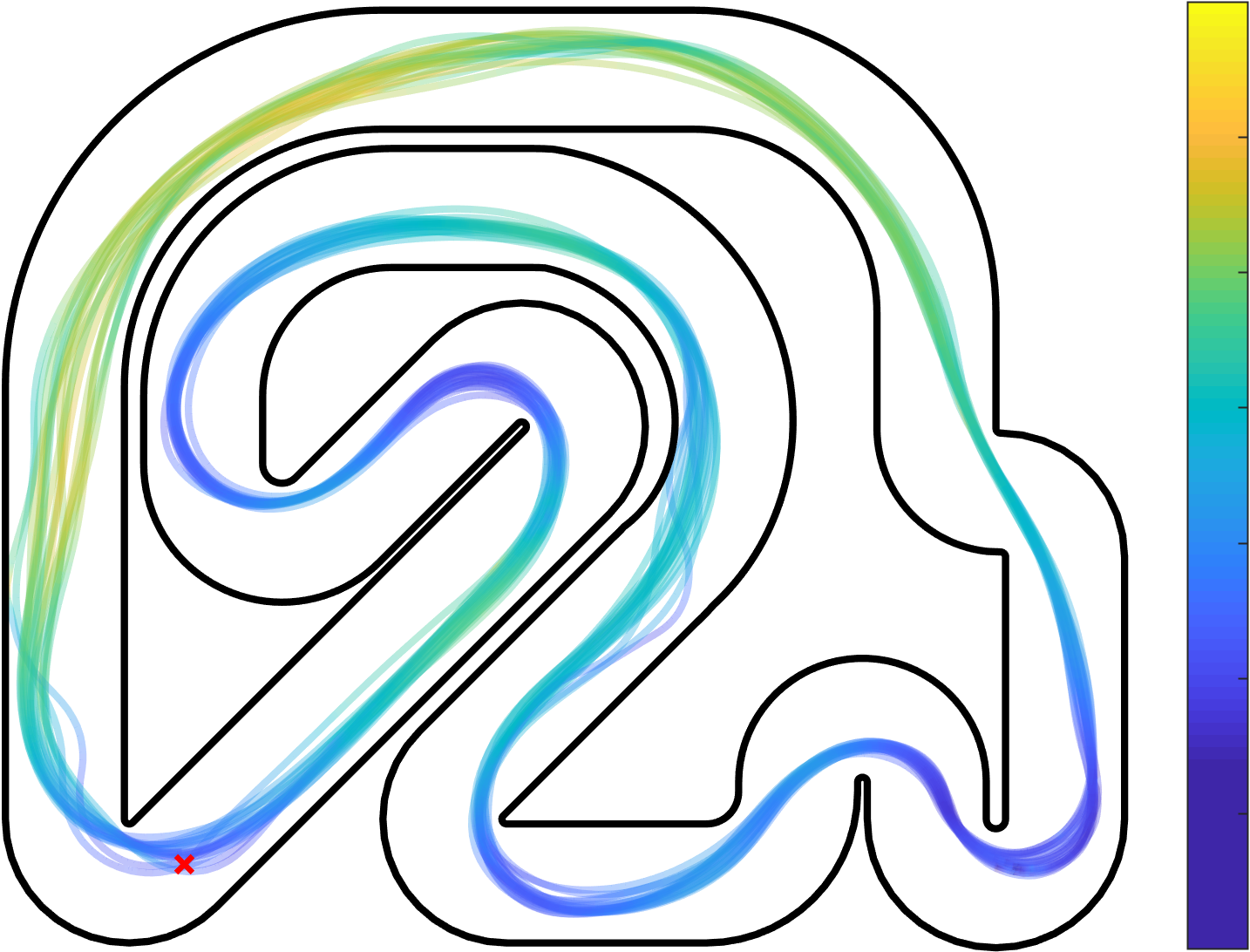}};
			\node[right =0.7cm of pic, rotate=90, xshift = -1.2cm] {Velocity [$\si{m/s}$]};
			\node[right =-0.1cm of pic, yshift = -2.45cm] {$0$};
			\node[right =-0.1cm of pic, yshift = -1.05cm] {$1$};
			\node[right =-0.1cm of pic, yshift = 0.38cm] {$2$};
			\node[right =-0.1cm of pic, yshift = 1.8cm] {$3$};
		\end{tikzpicture}
		\caption{GP-based controller}\label{fg:ORCA_gp}
	\end{subfigure}
\caption{Comparison of racelines with nominal and GP-based controller. The color indicates the 2-norm of the velocity, the red cross is the starting point of the race car.}
\end{figure*}
%==================================================
\subsection{GP-based Racing Controller}\label{ssc:ContouringControl}
%==================================================
%
We consider a race track given by its centerline and a fixed track width. The
centerline is described by a piecewise cubic spline polynomial, which is
parameterized by the path length $\Theta$. Progress along the track is
characterized by $\Theta_i$, which is introduced as an additional state and
enters the considered cost function linearly, encouraging a maximization of
progress along the track. Given a $\Theta_i$, we can evaluate the corresponding
centerline position $C(\Theta) = {[X_c(\Theta),Y_c(\Theta)]}^\tp$, such that the
constraint for the car to stay within the track boundaries can be expressed as
\begin{equation*}
	\mathcal{X}(\Theta_i) := \left\{x^{XY} \, \middle| \, \left \Vert x^{XY} - C(\Theta_i) \right\Vert \leq r \right\} \subset \mathbb{R}^2\eqc
\end{equation*}
where $r$ is half the track width. 

%==================================================
%\subsection{GP-based Controller Design}
%\label{sc:simulation_setup}
%%==================================================
Based on approximate Gaussian distributions of the state in prediction we have
for the position error $e^{XY}_i = x^{XY}_i - \mu^{XY}_i \sim
	\mathcal{N}(0, \Sigma^{XY}_i)$ and define an ellipsoidal PRS as
\[ \set{R}^{ell}(\Sigma^{XY}_i) = \setst{e^{XY}}{ (e^{XY})^\tp {(\Sigma^{XY}_i)}^{-1} e^{XY} \leq \chi^2_2(p_x)} \eqc\]
where $\chi^2_2(p_x)$ is the quantile function of the chi-squared distribution
with two degrees of freedom. 
In order to simplify constraint tightening of $\set{X}(\Theta_i)$ we find an
outer approximation by a ball $\set{R}^{b} \subseteq \set{R}^{ell}$ as 
\[ \set{R}^{b}(\Sigma^{XY}_i) = \setst{e^{XY}}{ \left\Vert e^{XY} \right\Vert \leq \sqrt{
		\lambda_{\max}\!\left(\Sigma^{XY}_i\right)\chi^2_2(p_x)}} \eqc \]
where $\lambda_{\max}\!\left(\cdot\right)$ is the maximum eigenvalue which can be
readily computed since $\Sigma^{XY}_i$ is a 2 by 2 matrix.
The necessary
constraint tightening can therefore be expressed as
\begin{align} \label{eq:raceCar_constrTightening}
	\set{Z}(\Theta_i,\Sigma^{XY}_i) &= \set{X}(\Theta_i) \ominus \set{R}^{b}(\Sigma^{XY}_i) \\ &= \setst{z^{XY}}{\left\Vert z^{XY}  - C(\Theta_i) \right\Vert \leq \tilde{r}\left(\Sigma^{XY}_i\right)} \eqc  \nonumber
\end{align}
where $\tilde{r}\!\left(\Sigma^{XY}_i\right) = r - \sqrt{\chi^2_2(p_x)
		\lambda_{\max}\!\left(\Sigma_i^{XY}\right)}$.
Fig.~\ref{fg:raceCar_constrTightening} exemplifies the predicted evolution of
the cars position and the resulting constraint tightening.

% We record data with a controller based on only
% the nominal system model $f(x,u)$ and choosing a squared exponential
% kernel~\eqref{eq:SE_kernel} we maximize the marginal likelihood to tune the
% necessary hyperparameters as well as process noise level $\Sigma_w$. 
The state is extended by previously applied inputs and large input changes are
additionally penalized. We make use of the Taylor
approximation~\eqref{eq:taylor_approx} to propagate uncertainties without the
use of an ancillary state feedback controller, i.e. $K = 0$. The prediction
horizon is chosen as $N = 30$ and we formulate the chance
constraints~\eqref{eq:raceCar_constrTightening} with $\chi^2_2(p_x) = 1$. To
reduce conservatism of the controller, constraints are instead only tightened
for the first $20$ prediction steps and are applied to the mean for the
remainder of the prediction horizon, similar to the method used
in~\cite{Carrau2016}. We reduce computation times by making use of the dynamic
sparse approximations with $10$ inducing points as outlined in
Section~\ref{ssc:DynamicSparseGP}, placing the inducing inputs regularly along
the previous solution trajectory.

To ensure real-time feasibility of the approach for the sampling time of $20
\text{ms}$ two additional approximations are applied. The variance dynamics are
pre-evaluated based on the previous MPC solution, which enables the
pre-computation of state constraints~\eqref{eq:raceCar_constrTightening} such
that they remain fixed during optimization. Additionally, we neglect the mean
prediction of the lateral velocity error $v_y$ since the state is difficult to
estimate reliably and the error generally small, such that we observed no
improvement when including it in the control formulation.
%
%==================================================
\subsection{Results}\label{sc:simulation_setup}
%==================================================
%
We start out racing the car using the nominal controller, i.e. the controller
without an added GP term, which therefore does not consider
uncertainties for constraint tightening. Since the nominal model is not well
tuned, driving behavior is
somewhat erratic and there are a number of small collisions with the track
boundaries. This can be observed in the racelines plotted in
Fig.~\ref{fg:ORCA_nominal} showing 20 laps run with the nominal controller.
Using the collected data, we train the GP error model $d$ using 325 data points.
We infer the hyperparameters in \eqref{eq:SE_kernel} as well as the noise level
$\Sigma_w$ using maximum likelihood optimization. 
% To further reduce computational effort we
% make use of the dynamic sparse GP approximation as outlined in
% Section~\ref{ssc:DynamicSparseGP} with 10 inducing points placed along the
% previous solution trajectory. We place inducing points more densly at the
% beginning of the trajectory to put extra emphasis on the current and near future
% states. 
The resulting racelines of 20 laps with the GP-based controller are
displayed in Figure~\ref{fg:ORCA_gp}, generally showing a much more
consistent and safe racing behavior. In particular, it can be seen that almost
all of the systematic and persistent problems in the raceline of the nominal
controller can be alleviated.

% \begin{figure}
% 	\center{}
% 	\begin{subfigure}{0.45\textwidth}
% 	\setlength\figureheight{3cm}
% 	\setlength\figurewidth{7cm}
% 	\input{figures/GPMPC_error.tex}
% 	\caption{Dynamic sparse GP compensation of the yaw-rate error with 10 inducing inputs during the first race lap. The black dots show the measured error on the yaw rate at each time step, while the blue line shows the error predicted by the GP. The shaded region is the 2-$\sigma$ confidence interval.}
% 	\label{fg:ORCA_error}
% 	\end{subfigure}

% 	\begin{subfigure}{0.45\textwidth}
% 		\setlength\figureheight{3cm}
% 		\setlength\figurewidth{7cm}
% 		\input{figures/GPMPC_errorhist.tex}
% 		\caption{Dynamic sparse GP compensation of the yaw-rate error with 10 inducing inputs during the first race lap. The black dots show the measured error on the yaw rate at each time step, while the blue line shows the error predicted by the GP. The shaded region is the 2-$\sigma$ confidence interval.}
% 		\label{fg:ORCA_error}
% 	\end{subfigure}
% \end{figure}

\begin{figure}[b]
	\center{}
	\setlength\figureheight{3cm}
	\setlength\figurewidth{7cm}
	\input{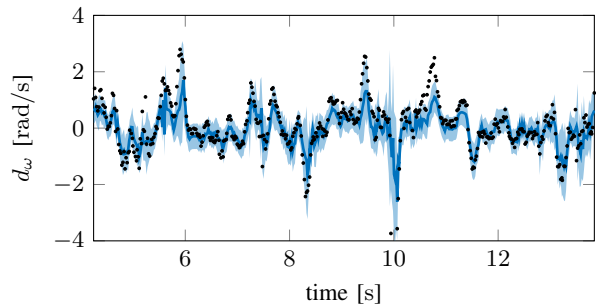}
	\caption{Dynamic sparse GP compensation of the yaw-rate error with 10 inducing inputs during the first race lap. The black dots show the measured error on the yaw rate at each time step, while the blue line shows the error predicted by the GP. The shaded region is the 2-$\sigma$ confidence interval.}\label{fg:ORCA_error}
\end{figure}

Fig.~\ref{fg:ORCA_error} shows the encountered dynamics error in the
yaw-rate and the predicted error during the first lap with the sparse
GP-based controller. Mean and residual uncertainty
predicted by the GP matches the encountered errors well. It is important to note
that the apparent volatility in the plot is not due to overfitting, 
but is instead due to fast changes in the input and matches the validation data,
i.e. the measured errors.
To quantify performance of the proposed controllers we compare in Table~\ref{tb:res} average lap time
$\overline{T}_l$, minimum lap time $T_{l,\min}$ as well as
the average 2-norm error in the system dynamics $\overline{\Vert e
\Vert}$, i.e.\ the difference between the mean state after one prediction step
and the realized state, $e(k\!+\!1) = \mu^x_1 - x(k\!+\!1)$.
We see that the GP-based controller is able to improve significantly on all
these quantities, with an average lap time improvement of $0.71$ s, or almost
$7\%$, which constitutes a large improvement in the considered racing task. This
is in part due to the improved system model, as evident in the average
dynamics error $\overline{\Vert e \Vert}$, but also due to the cautious nature
of the controller, which helps to further reduce collisions and large problems
in the raceline. Due to the cautious nature, the minimum
lap time gains are slightly less pronounced. In fact, the nominal controller
consistently displays higher top speeds, which often times, however, lead to
significant problems at the breakpoint to a slow corner.
Computation times are reported as average solve times $\overline{T}_c$ and 
the percentage of solutions in under 20 ms, $T_c < 20$ ms. Average solution
times of nominal and GP-based controller are similar. The percentiles
of solutions in under $20$ ms, however, differ significantly. This is mainly due
to frequent large re-planning actions occurring with the nominal controller.

The results therefore demonstrate that the presented GP-based controller can
significantly improve performance while maintaining safety in a hardware
implementation of a complex system with small sampling times.
\begin{table}
	\caption{Experimental results}\label{tb:res}
	\begin{center}
			   \begin{tabular}{cccccc} \toprule
				   Controller & $\overline{T}_l$ [$\si{s}$] & $T_{l,\min}$ [$\si{s}$]& $\overline{\Vert e \Vert}$ [$\si{-}$] & $\overline{T}_{\!c}$ [$\si{ms}$] & $T_c < \SI{20}{ms}$ \\  \midrule
				   Nominal  & 10.32 & 9.65 & 0.73 & 18.2 & $76.3 \%$ \\
				   GP-based & 9.61 & 9.27 & 0.33 & 17.2 & $99.8 \%$ \\
			   \end{tabular}
		   \end{center}
   \end{table}
% !TeX root = gpmpc.tex
%%%%%%%%%%%%%%%%%%%%%%%%%%%%%%%%%%%%%%%%%%%%%%%%%%%%%%%%%%%%%%%%%%%%%%%%%%%%%%%%
\section{Conclusion}\label{sc:conclusion}
%%%%%%%%%%%%%%%%%%%%%%%%%%%%%%%%%%%%%%%%%%%%%%%%%%%%%%%%%%%%%%%%%%%%%%%%%%%%%%%%
%
The paper discussed the use of Gaussian process regression to learn
nonlinearities for improved performance in model predictive control. Combining
GP dynamics with a nominal system allows for learning only parts of the
dynamics, which is key for keeping the required number of data points and
computational complexity of the GP regression feasible for online control.
Approximation methods for the propagation of the state distributions over the
prediction horizon were reviewed and it was shown how this enables a principled
treatment of chance constraints on both states and inputs.

A simulation example as well as a hardware implementation have shown how
the proposed formulations provide cautious control with improved performance
for medium-sized systems with low sampling times. In particular, we have
demonstrated in experiment that both performance and safety in an autonomous racing setting
can be significantly improved by using cautious data-driven techniques. 

\section*{Acknowledgement}
We would like to thank the Automatic Control Laboratory (IfA) at ETH Zurich and
in particular Alexander Liniger for his valuable input and support
with the hardware implementation on the miniature race cars.

%\appendices
%\input{06_appendix.tex}

\bibliographystyle{IEEEtran}
\bibliography{IEEEabrv,bibliography} %local

\end{document}